%% file: projections-infinitary-rewriting-lombardi-rios-devrijer.tex
\documentclass{entcs} 
\usepackage{entcsmacro}
\usepackage{graphicx}

\usepackage{carlos-draft}
\usepackage{ordinals}
\usepackage{rewriting}
\usepackage{proof-terms}
\usepackage{peqence}

\usepackage{xcolor}

\newcommand{\weg}[1]{}
\newcommand{\carlosweg}[1]{}

\definecolor{DarkGreen}{rgb}{0 0.5 0}

\newcommand{\arraysep}{\\[-6pt]}

\newcommand{\textsep}{\vspace{-4pt}}
\newcommand{\otherminitemsep}{\vspace{-4pt}}

\newcommand{\newlineifdraft}{}

\newcommand{\esrsfn}{\mathsf{ers}}
\newcommand{\esrs}[1]{\esrsfn(#1)}

\newcommand{\insertrsfn}{\mathsf{irs}}
\newcommand{\insertrs}[3]{\insertrsfn(#1,#2,#3)}

\newcommand{\res}{/}
\newcommand{\wideres}{\,/\,}

\newcommand{\cfpceqfn}{\triangleright}
\newcommand{\cfpceq}[2]{#1 \cfpceqfn #2}

\def\lastname{Lombardi, R\'ios, de Vrijer}
\begin{document}


\begin{frontmatter}
  \title{Projections for infinitary rewriting} 
  \author{Carlos Lombardi}
  \vspace*{-2mm}
  \address{Universidad Nacional de Quilmes --
    Argentina --
		\texttt{carlos.lombardi@unq.edu.ar}} 
  \author{Alejandro R\'ios}
  \vspace*{-2mm}
  \address{Universidad de Buenos Aires --
    Argentina --
		\texttt{rios@dc.uba.ar}} 
  \author{Roel de Vrijer}
  \vspace*{-2mm}
  \address{Vrije Universiteit Amsterdam --
    The Netherlands --
		\texttt{r.c.de.vrijer@vu.nl}} 

\begin{abstract} 
\input{abstract}
\end{abstract}

\begin{keyword}
  infinitary term rewriting, proof terms, permutation equivalence, projection
\end{keyword}
\end{frontmatter}

\input{article-body}

\bibliographystyle{plain}
\bibliography{somebib}

\newpage
\input{appendixes}

\end{document}

%% file: abstract.tex
Proof terms in term rewriting are a representation means for reduction sequences, and more in general for contraction activity, allowing to distinguish \eg\ simultaneous from sequential reduction. 
Proof terms for finitary, first-order, left-linear term rewriting are described in \cite{terese}, ch.~8.
In a previous work \cite{nosotros-rta14} we defined an extension of the finitary proof-term formalism, that allows to describe contractions in \emph{infinitary} first-order term rewriting, and gave a characterisation of \peqence. 

In this work, we discuss how \emph{projections} of possibly infinite rewrite sequences can be modeled using proof terms. 
Again, the foundation is a characterisation of projections for finitary rewriting described in \cite{terese}, Sec.~8.7. 
We extend this characterisation to infinitary rewriting and also refine it, by
describing precisely the role that structural equivalence plays in the development of the notion of projection.
The characterisation we propose yields a definite expression, 
\ie\ a proof term, that describes the projection of an infinitary reduction over another.

To illustrate the working of projections, we show how a common reduct of a 
(possibly infinite) reduction and a single step that makes part of it can be obtained
via their respective projections. 
We show, by means of several examples, that the proposed definition yields the 
expected behavior also in cases beyond those covered by this result. 
Finally, we discuss how the notion of limit is used in our definition of projection for infinite reduction.

%% file: article-body.tex
\section{Introduction}
\label{sec:introduction}

\input{introduction}

\section{Preliminaries}
\label{sec:preliminaries}

\input{proof-terms-intro}

%

\section{Finitary and infinitary projections}
\label{sec:projection-intro}

\input{projection-intro}

\newpage
\section{Projection through proof terms}
\label{sec:projection}

\input{projection}

%

\section{A partial confluence property}
\label{sec:property}

\input{property}

\vspace*{-3mm}
\section{Limitations of this approach}
\label{sec:limitations}

\vspace*{-1mm}
\input{limitations}


\section{Conclusions and future research directions}
\label{sec:conclusions}

\vspace*{-1mm}
\input{conclusions}

\vspace*{-1mm}

%% file: introduction.tex

The general scope of this article is infinitary, first-order, left-linear term rewriting,
with strong convergence as the criterion for limits of infinite reductions.
 
The same principles and notions used to study sequences of numbers, or more generally, of points in a topological space, can be applied to \emph{\redseqs} (which are sequences of rewriting steps), and particularly to \emph{infinite} ones.
By adapting the notions of limit and convergence, a target term can be determined for some infinite sequences. Such targets are, usually, infinite terms.

The possibility of infinite \redseqs\ having targets leads to the realm of \emph{infinitary rewriting}. It is natural to wonder whether the notions and results known for finite rewriting have extensions in the infinitary setting. 
Several results, both positive and negative,
appear in the literature of the last 25 years
\cite{rewriterewrite}, \cite{orthogonal-itrs-95}, \cite{terese}.

The notion of \emph{projecting} a \redseq\ over another, coinitial, one, has been extensively studied
and finds its origin in a key lemma for confluence of lambda calculus and orthogonal term rewriting, 
the Parallel Moves Lemma
\cite{curry-feys}, \cite{barendregt}, \cite{terese}.
Projections may be used to formulate stronger versions of  \emph{confluence}%
\footnote{The study of infinitary confluence is not a mere extension of the results known for the finitary case. \Eg, the infinitary counterpart of the Newman lemma does not hold, \confer\ \cite{inf-ars,inf-normalization}.}%
.
Given two coinitial sequences $\reda$ and $\redb$, where $t \sredx{\reda} s$ and $t \sredx{\redb} u$, a common reduct of $s$ and $u$ can be obtained by applying to them the projection of $\redb$ over $\reda$, and that of $\reda$ over $\redb$, respectively.
This statement can be further strengthened using characterisations of \emph{permutation equivalence} of reductions. If we use the notation $\reda \res \redb$ for the projection of $\reda$ over $\redb$, $\peq$ for  permutation equivalence and an infix colon ;  for concatenation of reductions,  then a stronger variant of confluence can be stated as follows:
\vspace*{-2mm}
$$ \reda \,;\, \redb \res \reda \ \peq \ \redb \,;\, \reda \res \redb$$
\vspace*{-8mm}

The aim of this article is to present some preliminary definitions and results related to projections, taken  from our ongoing work on infinitary permutation equivalence. 
\weg{%
The general scope of our work, and thereby of this article, is infinitary, first-order, left-linear term rewriting.%
}%
More in particular, our goal is to define projection in such a way that an explicit expression is obtained, representing the projection of an infinitary reduction over another.
We also want to find out to which extent such a characterisation involves the notion of limit.

To this end, we use the representation of infinitary rewriting by means of \emph{proof terms} given in \cite{nosotros-rta14}, which extends that given for finitary, first-order, left-linear term rewriting in \cite{terese}.
A proof term is an expression, namely a term, that describes a reduction. As a matter of fact, something more general: any combination of simultaneous (\ie\ multistep) and sequential reduction can be denoted by a proof term.
\weg{A proof term is an expression, namely a term, that describes reduction, in particular 
any combination of simultaneous (\ie\ multistep) and sequential reduction. }%
Composition, or concatenation, of reductions is represented in the proof term formalism by a binary symbol.  An infix dot is used, so that the composition of (the reductions denoted by) the proof terms $\psi$ and $\phi$ is noted $\psi \comp \phi$.
Infinitary \peqence\ is modeled by equational logic applied to proof terms. 

The study of equivalence between reductions in \cite{terese} includes a characterisation of  projection of  one reduction over another, by means of the binary operation $\res$ defined between proof terms. 
That is, if $\psi$ and $\phi$ are proof terms, then $\psi \res \phi$ is a proof term that represents the projection of $\psi$ over $\phi$.
The definition of the projection operation is given modulo structural equivalence, a subrelation of \peqence\ that is specific for the proof term formalism.  Therefore, some details about how to obtain the proof term corresponding to a projection are left open in that definition.

\medskip
\noindent
\textbf{Results and discussion}

We give a definition of projections for infinitary rewriting, which extends and refines that given in \cite{terese} for the finitary case. The refinement consists in specifying some of the permutation-equivalence transformations that are needed in order to compute projections.

We show a partial confluence result about this definition.
Given a (possibly infinite) reduction $\psi$ and one of its constituent steps, let us call it $\phi$, such that the step can be performed on the source of $\psi$, we prove that $\psi \comp (\phi \res \psi) \peq \phi \comp (\psi \res \phi)$. This statement corresponds to the strengthened variant of confluence described earlier, as expressed by means of proof terms.
We prove this result not in full generality.  The minimal requirement on the step $\phi$ would be that it can already be performed in the source term of the proof term $\psi$, that is, that it does not depend on any previous step in the reduction represented by $\psi$.
This requirement is strengthened in the sense that it not only \emph{holds} for $\phi$, 
but that, moreover, this is in some sense \emph{evident}, just from the syntactic form of the proof term $\psi$.%

This restriction, made specifically for this exploratory paper, has a twofold motivation. 
Firstly, 
it keeps matters simple, so that they can be clearly explained. Generalisations can be obtained, but they require more complicated techniques. 
Secondly, 
it turns out that in our work on infinitary standardisation, a major motivation for our interest in projections, nothing more is needed.

We show that our definition behaves as expected in some cases that extend the scope of the proven property, by means of several examples.
We remark that in many cases the computation of the (proof term representing the) projection uses the notion of limit only to obtain the source or target term of a proof term; limits are not needed in order to reason specifically about projections. This includes computations of the projection of an infinite reduction over a finite one, and conversely, of a finite reduction over an infinite one.

We point out that limits are needed though, in some finite-over-infinite cases, related to \emph{infinitary erasure}, and also to compute infinite-over-infinite projections. 

\medskip\noindent
\textbf{Structure of the paper}

In Section~\ref{sec:preliminaries}, we give the needed definitions about infinitary rewriting and the proof term model.
After a preliminary discussion in Section~\ref{sec:projection-intro}, we introduce the definition of projection in Section~\ref{sec:projection}, analyzing it through several examples, and we state and prove our partial confluence result in Section~\ref{sec:property}.
In Section~\ref{sec:limitations}, we explore cases where the explicit mention of  limits in the definition of projection cannot longer be avoided.
Finally, some preliminary conclusions of this work-in-progress, and possible directions for future research, are given in Section~\ref{sec:conclusions}.
An extended version \cite{long-version} includes the omitted proofs, and also some additional material regarding the formal definition of projections.
\vspace{-1ex}

%% file: proof-terms-intro.tex
We briefly introduce infinitary rewriting by means of the TRS with signature $\set{a/0, f/1, g/1, k/1}$ and the rules  $f(x) \to g(x), g(x) \to k(x)$. Consider the term $f^n(a)$ for some $n < \omega$. In the tree rendering of this term, a sequence of $n$ occurrences of $f$ precedes the  occurrence of $a$.
An \emph{infinite} sequence of chained $f$ symbols represents an \emph{infinite term}, 
which we denote  as $f\om$. 
For each $n < \omega$, this linear tree has an occurrence of $f$ at depth $n$. 
This term is the source of the \emph{infinite reduction sequence} 
$f\om \to g(f\om) \to g^2(f\om) \ldots g^n(f\om) \to g^{n+1}(f\om) \ldots \ $. 
Note that the infinite sequence formed by the \emph{targets} of the successive prefixes of this reduction, namely $\langle g(f\om), g^2(f\om), \ldots g^n(f\om) \ldots \rangle$, converges with $g\om$ as limit.
Additionally, the sequence given by the \emph{depth} (distance to the root) in which each step is performed, is simply $\langle 0, 1, 2, \ldots n \ldots \rangle$, so that it tends to infinity.
Such a reduction sequence is considered as (strongly) convergent, having $g\om$ as target.
In turn, it can be further extended from this target, leading to the following reduction sequence
$f\om \to g(f\om) \to g^2(f\om) \to \ldots \, g\om \to k(g\om) \to k^2(g\om) \to \ldots \, k\om$, whose length is $\omega * 2$.

These simple examples show that the application of the notions of \emph{limit} and \emph{convergence} to the study of reduction sequences, lies in the foundation of infinitary term rewriting.

Infinitary term rewriting allows to rigorously define infinite terms and convergent infinite reductions, and study their properties. We refer to Chapter~12 in \cite{terese} and to \cite{inf-normalization} for the basic definitions.  Here we just remark that we adopt
the \emph{strong convergence} criterion:  for a transfinite rewrite sequence to be convergent, we require 
the depths of the successive steps to tend to infinity at each limit ordinal. 

Projections of possibly infinite reductions are also defined in \cite{orthogonal-itrs-95}, and in a similar way, in \cite{terese}, Chapter~12; our work proposes an alternative approach to that subject, via proof terms. Proof terms for term rewriting were introduced in  \cite{terese}, Chapter~8 and have been adapted to the infinitary setting in \cite{nosotros-rta14} and \cite{phdcarlos}.

The idea motivating the definition and application of proof terms is to denote the reductions of some calculus as \emph{terms} over an extended signature.
For each reduction rule in the original TRS, a \emph{rule symbol} is introduced.
The arity of a rule symbol coincides with the number of different variables occurring in the left-hand side of the rule it represents.
\Eg, the signature of proof terms for a first-order TRS \trst\ including the rules $f(x) \to g(x)$, $j(m(x),m(y)) \to k(x)$ and $g(x) \to k(x)$ adds the rule symbols $\mu/1$, $\rho/2$ and $\nu/1$. 

The initial stage in the definition of infinitary proof terms, as given in \cite{phdcarlos,nosotros-rta14}, is the set of \emph{infinitary multi-steps}, \ie, the finite or infinite terms over the signature extended with rule symbols.
Multi-steps with exactly one occurrence of a rule symbol denote single reduction steps, \eg\ $\mu(a) : f(a) \to g(a)$, $g(\rho(a,b)) : g(j(m(a),m(b))) \to g(k(a))$. We identify such proof terms as \emph{one-steps}. With more occurrences of rule symbols, we denote multi-steps, like $j(\mu(a), \mu(b)) : j(f(a),f(b)) \sred j(g(a),g(b))$, $\rho(\mu(a),b) : j(m(f(a)),m(b)) \sred k(g(a))$.
A multi-step can be infinite, and even contain infinitely many
rule symbol occurrences, as \eg\ 
$\mu^\omega: f^\omega \infred g^\omega$. 

The beginning and end terms of the corresponding reductions are called the
source and target of the proof term. For the proof terms considered so far,
they can be  obtained via rewriting 
in two companion \TRSs, denoted as $SRC$ and $TGT$ respectively.
For each rule symbol $\rho : l \to r$, $SRC$ includes a rule $\rho(x_1, \ldots, x_m) \to l[x_1, \ldots, x_m]$ 
and $TGT$  a rule $\rho(x_1, \ldots, x_m) \to r[x_1, \ldots, x_m]$.  
Source and target  of a proof term are its normal forms in  $SRC$ and $TGT$, respectively.
Of course there are the questions of existence and uniqueness.
First note that both $SRC$ and $TGT$ have unique normal forms, since they are orthogonal infinitary TRSs. It is also not hard to verify that $SRC$ enjoys infinitary strong normalisation ($SN^\infty$). Contrarily, $TGT$ does not enjoy even infinitary weak normalisation ($WN^\infty$) if the TRS includes collapsing rules. 
We conclude that the source of an infinitary multi-step $\psi$ is always uniquely defined.  
The target is only defined if $\psi$ is $WN^\infty$, but if so, it is also unique.
If $\psi$ is not $WN^\infty$ for $TGT$, then we say that $tgt(\psi)$ is undefined.

The set of redexes in $src(\psi)$ corresponding to the rule symbol occurrences in 
$\psi$ admits at least one convergent development (respectively, all 
developments are convergent) 
precisely if  $\psi$ is $WN^\infty$ (respectively $SN^\infty$) in the \TRS\ $TGT$.
An \imstep\ is called \emph{convergent}, if its target can be computed. 

To complete the definition of the set of finitary proof terms, 
we add a new binary function symbol $\comp$ (written infix), expressing concatenation, 
or composition, of reductions.
Just to give a simple example, the proof term $f(\mu(a)) \comp f(\nu(a))$ 
denotes the two-step reduction $f(f(a)) \to f(g(a)) \to f(k(a))$. The same reduction 
is represented by the proof term $f(\mu(a) \comp \nu(a))$.
Not all terms over the thus 
 extended signature are valid proof terms though, but only those that can be constructed starting
 from the infinitary multi-steps
 by  the following three inductive clauses.
 
First, closure under function or rule symbols:
if $\psi_1, \ldots, \psi_n$ are proof-terms,  then so are 
$f(\psi_1, \ldots, \psi_n)$ and $\mu(\psi_1, \ldots, \psi_n)$.
Source and target terms are defined as expected, \eg\ $src(\mu(\psi_1, \ldots, \psi_n)) = l[src(\psi_1), \ldots, src(\psi_n)]$, where $\mu: l \to h$.
 
Secondly, \emph{binary composition}: 
if $\psi, \phi$ are proof terms, then so is $\psi \comp \phi$,  
provided that $tgt(\psi) = src(\phi)$.  This presupposes convergence of $\psi$.
The proof term $\psi \comp \phi$ is convergent iff $\phi$ is.
We define $src(\psi \comp \phi) = src(\psi)$ and $tgt(\psi \comp \phi) = tgt(\phi)$.

\noindent
\begin{tabular}{@{}p{91mm}c}
\begin{minipage}{91mm}
\normalsize{
\hspace{1.1em}
Thirdly, \emph{infinite composition}: the term corresponding to the figure 
is a proof term, 
if $\psi_0, \psi_1, \psi_2, \ldots$ are, 
provided that
for each $i < \omega$ we have convergence of $\psi_i$  and $tgt(\psi_{i}) = src(\psi_{i+1})$. 
A linear rendering  would be $\psi_0 \comp (\psi_1 \comp (\psi_2 \comp \ldots ))$.
We use $\icomp \psi_i$ as shorthand for this proof term.
}
\end{minipage}
&
\begin{minipage}{38mm}
\hspace*{\stretch{1}}
$\xymatrix@C-20pt@R-20pt{
& \ar[dl] \ar[dr] \cdot \\
\psi_0 & & \ar[dl] \ar[dr] \cdot \\
& \psi_1 & & \ar[dl] \ar[dr] \cdot \\
& & \psi_2 & & \ddots \\
}$
\hspace*{\stretch{1}}
\end{minipage}
\end{tabular}

For $\psi = \icomp \psi_i$, we define $src(\psi) = src(\psi_0)$, and declare that $\psi$ is convergent iff the sequence $\langle \mind{\psi_i} \rangle_{i<\omega}$ tends to infinity. Here $\mindfn$ stands for the \emph{minimal activity depth} of a proof term. \Eg, $\mind{f(\mu(a)) \comp \mu(g(a))} = 0$, since the denoted activity includes a root step; while $\mind{m(f(\mu(a))) \comp m(\mu(g(a)))} = 1$ as the denoted steps are at depths 2 and 1 resp.. \Confer\ \cite{phdcarlos} for details.
If $\icomp \psi_i$ is convergent, then $tgt(\psi)$ is defined as the limit of the sequence $\langle tgt(\psi_i) \rangle_{i<\omega}$.

By the above definition of proof terms an infinite composition is also a binary composition:
$\icomp \psi_i = \psi_0 \comp (\icomp \psi_{i+1})$. 
To preserve unique constructibility, infinitary proof terms are defined in \cite{phdcarlos} in layers corresponding to ordinal numbers, such that each proof term has a unique layer. Particularly, the (unique) layer of $\psi$ is a limit ordinal iff $\psi$ is an infinite composition.

\emph{\Peqence} (noted $\peq$ henceforth) relates the proof terms that denote the same reduction in different ways, regarding parallelism/nesting degree, sequential order, and/or localisation.
This relation is defined, in \cite{phdcarlos,nosotros-rta14}, as the congruence generated by 
the following seven basic equivalences: 

$
\begin{array}{lrcl}
\peqidleft & src(\psi) \comp \psi & \eqnpeq & \psi \arraysep
\peqidright & \psi \comp tgt(\psi) & \eqnpeq & \psi \arraysep
\peqassoc & \psi \comp (\phi \comp \chi) 
	& \eqnpeq & 
	(\psi \comp \phi) \comp \chi \arraysep
\peqstruct & f(\psi_1, \ldots, \psi_m) \comp f(\phi_1, \ldots, \phi_m) 
	& \eqnpeq & 
	f(\psi_1 \comp \phi_1, \ldots, \psi_m \comp \phi_m) \arraysep
\peqinfstruct &
  \icomp f(\psi^1_i, \ldots, \psi^m_i)
	& \eqnpeq & 
  f(\icomp \psi^1_i, \ldots, \icomp \psi^m_i) \arraysep
\peqoutin & \mu(\psi_1, \ldots, \psi_m) 
	& \eqnpeq & 
	\mu(s_1, \ldots, s_m) \comp r[\psi_1, \ldots, \psi_m] \arraysep
\peqinout & \mu(\psi_1, \ldots, \psi_m) 
	& \eqnpeq & 
	l[\psi_1, \ldots, \psi_m] \comp \mu(t_1, \ldots, t_m) 
\end{array}
$

\noindent
where $\mu: l \to r$, $s_i = src(\psi_i)$ and $t_i = tgt(\psi_i)$ in \peqinout\ and \peqoutin, augmented with the following equational logic rules:

$
    \begin{array}{c} 
      \psi_i \peq \phi_i \quad \textforall i < \omega \\ 
	    \hline 
      \icomp \psi_i \ \peq\  \icomp \phi_i
    \end{array} 
    \ \ \eqlinfcomp
$

$
    \begin{array}{c} 
    \begin{array}{ll}
      \begin{array}{l}
      \textforall k < \omega \arraysep
      \textnormal{exists } \chi_k, \psi'_k, \phi'_k
      \end{array}
    &
      \left\{
      \begin{array}{l}
      \psi \peqe \chi_k \comp \psi'_k \quad \mind{\psi'_k} > k 
	    \arraysep
	    \phi \peqe \chi_k \comp \phi'_k \quad \mind{\phi'_k} > k 
	    
      \end{array}
      \right.
    \end{array}
    \\
    \hline 
    \psi \ \peq  \phi 
  \end{array} 
  \ \ \eqllim
$

\noindent
Here $\peqe$ is the congruence generated by the seven basic equations, augmented by \eqlinfcomp, but excluding the \eqllim-rule itself.

As a first example of proof terms including composition, and also of \peqence, we consider the proof terms $f \nu a \comp \mu k a$ and $\mu g a \comp g \nu a$ (we omit some unary symbol parentheses in the sequel). These proof terms represent the two possible reduction sequences that transform the source term $fga$ into $gka$. 
Note that \emph{simultaneous reduction} of a set of coinitial redexes is given, in the proof-term model, a specific denotation. In this case, the proof term $\mu \nu a$ denotes, specifically, the simultaneous step $f g a \mulstep g k a$.

We prove that the three given proof terms are \peqent, as follows.
By \peqinout\ and \peqoutin\ we obtain $\mu \nu a \peq f \nu a \comp \mu k a$ and $\mu \nu a \peq \mu g a \comp g \nu a$. Symmetry and transitivity, which are included in the generated congruence, yield $f \nu a \comp \mu k a \peq \mu g a \comp g \nu a$.
Note that the \peqinout\ and \peqoutin\ equations model the permutation of a head step \wrt\ internal activity.

In turn, the \peqstruct\ equation allows to reason about activity lying inside a \emph{fixed context}, 
as in the following \peqence\ judgement:
$m f \nu a \comp m \mu k a 
	\peq m(f \nu a \comp \mu k a)
	\peq m \mu \nu a
	\peq m (\mu g a \comp g \nu a)
	\peq m \mu g a \comp m g \nu a
$ where the first use of \peqstruct\ enables the permutation of steps, and the second one yields the equivalence between reduction sequences. Here the fixed context is $m(\Box)$.


The next example involves infinite composition. 
Consider $\psi = \psi_1 \comp \psi_2$ where $\psi_1 = \icomp g^i(\mu(f\om))$ and $\psi_2 \eqdef \icomp k^i(\nu(g\om))$, and $\phi = \icomp \chi_i$ where $\chi_i = k^i(\mu(f\om) \comp \nu(f\om))$.
The proof terms $\psi$ and $\phi$ denote, respectively, the reduction sequences 
$f\om \to g f\om \to g^2 f\om \infred g\om
  \to k g\om \to k^2 g\om \infred k\om$
and
$f\om \to g f\om \to k f\om \to k g f\om \to k^2 f\om \infred k\om$, 
that are two different ways to perform the transformation of each occurrence of $f$ in $f\om$ to $g$ and subsequently to $k$, by means of the $\mu$- and $\nu$-rules respectively.

Using the augmented congruence, including the \eqllim\ rule, the assertion $\psi \peq \phi$ can be justified.
To start, note that $\psi_1 = \mu f\om \comp \icomp g(g^i \mu f\om)$ just by definition of infinite compositions. 
In turn, \peqinfstruct\ yields $\icomp g(g^i \mu f\om) \peqe g( \icomp g^i \mu f\om) = g(\psi_1)$.
Applying a similar argument on $\psi_2$, and then \peqassoc, we obtain 
$\psi \peqe \mu f\om  \comp (g(\psi_1) \comp \nu g\om ) \comp k(\psi_2)$.
Then, a permutation of steps based on \peqinout\ and \peqoutin\ yields
$\psi \peqe \mu f\om  \comp (\nu f\om \comp k (\psi_1) ) \comp k(\psi_2)$, so that we get
$\psi \peqe (\mu f\om \comp \nu f\om) \comp k(\psi_1 \comp \psi_2) = \chi_0 \comp k(\psi)$, by \peqassoc\ and \peqstruct.
For any $n < \omega$, iterating over the whole argument yields 
$\psi \peqe \chi_0 \comp \chi_1 \comp \ldots \comp \ldots \comp \chi_n \comp k^{n+1}(\psi)$.
On the other hand, it is straightforward to obtain $\phi \peqe \chi_0 \comp \chi_1 \comp \ldots \chi_n \comp \icomp \chi_{n+1+i}$.
Hence \eqllim\ yields $\psi \peq \phi$.

%
%

This example shows the relevance of the \eqllim\ rule for  \peqence\ judgements. 
Proof terms $\psi$ and $\phi$ can be proven $\peqe$-equivalent up to an arbitrary activity depth level $n$:  
we have $\psi \peqe \chi_0 \comp \ldots \comp \chi_n \comp \psi'$ and $\phi \peqe \chi_0 \comp \ldots \comp \chi_n \comp \phi'$, where $\mind{\psi'} > n$ and $\mind{\phi'} > n$.
So, $\psi$ and $\phi$ can be transformed into forms whose difference, represented by $\psi'$ and $\phi'$, can be made arbitrarily irrelevant (with minimum activity depth as  the relevance measure).
It is not possible to obtain $\psi \peqe \phi$, however. 
The \eqllim\ rule allows taking limits to conclude $\psi \peq \phi$.


Similarly, the infinite multistep $\mu\om$ and the infinite composition $\icomp g^n(\mu(f\om))$ denote, respectively, the simultaneous and sequential contraction of the infinite set of $\mu$-redexes present in the source term $f\om$. In fact, the latter corresponds 
to the \redseq\ $f\om \sstep g(f\om) \sstep g^2(f\om) \sstep \ldots \infred g\om$.
Other sequential reductions of the same set of redexes are denoted by specific infinite composition proof terms. As an example, the sequence
$f\om \sstep f(g(f\om) \sstep g^2(f\om) \sstep g^2(f(g(f\om)) \sstep g^4(f\om) \infred g\om$ can be faithfully denoted by $\icomp g^{2i}(f(\mu(f\om))) \comp g^{2i}(\mu(g(f\om)))$.
Again, all these proof terms can be proven \peqent. In the infinite case, the corresponding equivalence judgement makes use of the \eqllim\ equational rule.

In \cite{nosotros-rta14} we showed that any convergent \redseq\ can be given a precise denotation as a \emph{stepwise} proof term, \ie, a proof term constructed from one-steps, by only using binary and infinitary composition. Moreover, this representation is \emph{unique} modulo the associativity of the composition symbol.
Note that \eg\ $(\mu(f(a)) \cdot \nu(f(a))) \cdot k(\mu(a))$ and $\mu(f(a)) \cdot (\nu(f(a)) \cdot k(\mu(a)))$ are different, albeit equivalent, proof terms. 


We gave in \cite{nosotros-rta14} also an alternative proof of the \emph{compression} property for convergent transfinite rewrite sequences, using their representations as proof terms. 
In fact, we proved a strong version: the compressed (\ie\ having length at most $\omega$) \redseq\ is \peqent\ (and not only coincident in source and target) to the original one. 
\weg{We also pointed out that orthogonality is not required at the TRS level.}%
The general argument of our compression proof reflects a remark in \cite{KetemaRTA12}: compression can be considered as a degenerate form of \emph{standardisation}. Based on this idea, we are currently working on standardisation results for infinitary rewriting, also based on the representation of reductions
by means of proof terms.

Finally, we remark that our definition of the set of proof terms, as well as our characterisation of \peqence, are based on \emph{inductive} notions and techniques.
In particular, inductive reasoning can be used on the set of occurrences in a term, considering their distance to the root which is always finite.
Also, transfinite induction can be used to reason about infinite reduction sequences, since their length can always be expressed as an ordinal.

An alternative approach that incorporates \emph{coinductive} techniques, appears in \cite{endrullis-rta-2015}.
There, convergent reduction sequences are represented by coinductively defined trees, and the \emph{reduction relation} is characterised through a combination of inductive and coinductive fixed points. 
The latter characterisation is formalised in Coq, leading to a Coq-certified proof of compression. The approach is also extended to study infinitary equational reasoning.
On the other hand, their proposal does not describe the space of transfinite reductions in full detail. In particular, it does not allow  different descriptions of sequential and simultaneous reduction, and the order in which disjoint steps are performed cannot be expressed. Neither \peqence\ nor projections are addressed%
\footnote{A limited form of \peqence\ is currently being studied \cite{endrullis-communication-2016}.}%
.
Hence, we perceive this work to be complementary with our characterisation of infinitary rewriting.

%% file: projection-intro.tex
\newcommand{\arInfRR}{
\ar@{*{}*@{-}*{\hspace{-3.5mm}\scriptscriptstyle >\hspace{-1mm}>\hspace{-1mm}>}}[rr]
}%
Let $\delta$ be a reduction sequence, and $\gamma$ a coinitial step. The following commutation diagram describes the argument of the Parallel Moves Lemma (PML). \\[-10pt]
$
\begin{array}{l@{\ }l@{\ }l}
\begin{array}{@{}l@{}} \\ \\ \gamma \end{array}
&
\underbrace{\overbrace{
\xymatrix@R=16pt{
\ar@{>>}[d] \ar[r] & 
\ar@{>>}[d] \ar[r] & \ar@{>>}[d] \ar@{.}[r] & \ar@{>>}[d] \ar[r] & 
\ar@{>>}[d] \\
\ar@{>>}[r] & 
\ar@{>>}[r] & \ar@{.}[r] & \ar@{>>}[r] &
}
}^{\delta}}_{\delta \res \gamma}
&
\begin{array}{@{}l@{}} \\ \\ \gamma \res \delta \end{array}
\end{array}
$

This diagram establishes a \emph{particular confluence} property: a common target can be reached by performing $\gamma \res \delta$ and $\delta \res \gamma$, after $\delta$ and $\gamma$ respectively.

If $\delta$ is an \emph{infinite} \redseq, the diagram gets infinite as well: \\[-10pt]
$
\begin{array}{l@{\ }l@{\ }l}
\begin{array}{@{}l@{}} \\ \\ \gamma \end{array}
&
\underbrace{\overbrace{
\xymatrix@R=16pt{
\ar@{>>}[d] \ar[r] & 
\ar@{>>}[d] \ar[r] & \ar@{>>}[d] \ar@{.}[r] & 
\ar@{>>}[d] \ar[r] & \ar@{>>}[d] \ar@{.}[r] & 
\ar@{>>}[d] \arInfRR & & \ar@{>>}[d] \\
\ar@{>>}[r] & 
\ar@{>>}[r] & \ar@{.}[r] & 
\ar@{>>}[r] & \ar@{.}[r] & 
\arInfRR & &
}
}^{\delta}}_{\delta \res \gamma}
&
\begin{array}{@{}l@{}} \\ \\ \gamma \res \delta \end{array}
\end{array}
$

\noindent
leading to an infinite variant of PML.
A concrete example follows, using the rule $f(x) \to g(x)$ and omitting parentheses for unary symbols \\[-10pt]
$
\begin{array}{@{}l@{\ }l@{\ }l}
\begin{array}{@{}l@{}} \\[12pt] \gamma \end{array}
&
\underbrace{\overbrace{
\xymatrix@C=14pt@R=16pt{
f\om \ar[d] \ar[r] & 
g f\om \ar[d] \ar[r] & 
g f g f\om \ar[d] \ar[r] & 
g f g^2 f\om \ar[d] \ar@{.}[r] & 
g f g^n f\om \ar[d] \ar[r] & 
g f g^{n+1} f\om \ar[d] \ar@{.}[r] & 
\ar@{>>}[d] \arInfRR & & 
g f g\om \ar[d] 
\\
f g f\om \ar[r] & 
g^2 f\om \ar[r] & 
g^3 f\om \ar[r] & 
g^4 f\om \ar@{.}[r] & 
g^{n+2} f\om \ar[r] & 
g^{n+3} f\om \ar@{.}[r] & 
\arInfRR & &
g\om
}
}^{\delta}}_{\delta \res \gamma}
&
\begin{array}{@{}l@{}} \\ \\[12pt] \gamma \res \delta \end{array}
\end{array}
$

Note that each step in $\delta$ has a nonempty projection after (the respective projection of) $\gamma$.
The projection of $\redb$ after $\reda$ can be naturally defined as the \emph{limit} of the projections after its successive prefixes.
In turn, the projection of $\reda$ after $\redb$ can be defined as the limit of the projections of the successive prefixes of the former.
Observe that the notion of limit is relevant for the definition of projections, whenever infinite reductions are involved.

In the sequel, we define the projection of one reduction over another 
as a binary operation on (their representation as) \emph{proof terms}. 
As permutation equivalence is also characterized on proof terms, we can express in this formalism the stronger version of the confluence criterion suggested by the PML described in the introduction, as follows: $\psi \comp \phi \res \psi \ \peq \ \phi \comp \psi \res \phi$, where $\psi$ and $\phi$ represent $\delta$ and $\gamma$ resp..
This statement also expresses the idea of \emph{orthogonality} between $\psi$ and $\phi$ in a way independent from the syntax of terms, or more generally, the form of the objects being rewritten. As such, it is closely related to the axiom called PERM in \cite{thesis-mellies}, and  \emph{Semantic orthogonality} in \cite{nosotrosPopl2014}.

In the next section, we extend to the infinitary realm a definition given in \cite{terese}, showing that the role of limits in computing projections is very restricted in some cases, as in the example just given.
The strong confluence result is proved, for a very limited case, in Section~\ref{sec:property}.

%% file: projection.tex
The projection of a reduction over another is defined in \cite{terese} Ch.~8, for finitary term rewriting, as the operation on proof terms defined as follows.

$
\begin{array}{rcl}
\mu(\phi_1, \ldots, \phi_m) \wideres \mu(\psi_1, \ldots, \psi_m) 
	& = & h[\phi_1 \res \psi_1, \ldots, \phi_m \res \psi_m] \arraysep
\mu(\phi_1, \ldots, \phi_m) \wideres l[\psi_1, \ldots, \psi_m]
	& = & \mu(\phi_1 \res \psi_1, \ldots, \phi_m \res \psi_m) \arraysep
l[\phi_1, \ldots, \phi_m] \wideres \mu(\psi_1, \ldots, \psi_m)
	& = & h[\phi_1 \res \psi_1, \ldots, \phi_m \res \psi_m] \arraysep
f(\phi_1, \ldots, \phi_m) \wideres f(\psi_1, \ldots, \psi_m)
	& = & f(\phi_1 \res \psi_1, \ldots, \phi_m \res \psi_m) \arraysep
(\phi \comp \psi) \res \chi 
	& = & \phi \res \chi \comp \psi \wideres (\chi \res \phi) \arraysep
	\chi \res (\phi \comp \psi) 
	& = & (\chi \res \phi) \wideres \psi
\end{array}
$ \\[2pt]
where $\mu : l \to h$.
We give a simple example, using the
rules: $\rho: j(g(x),y) \to j(x,y)$, $\mu: f(x) \to g(x)$, $\pi: a \to b$, $\tau: c \to d$, $\sigma: m(x) \to n(x)$. \\[2pt]
$
\begin{array}{rcl}
\multicolumn{3}{l}{\big(j(\mu(\pi),m(c)) \comp \rho(b,\sigma(c))\big) \wideres j(f(\pi),\sigma(\tau))} \arraysep
	& = &  j(\mu(\pi),m(c)) \res j(f(\pi),\sigma(\tau)) 
			\comp 
			\rho(b,\sigma(c)) \wideres (j(f(\pi),\sigma(\tau)) \res j(\mu(\pi),m(c))) \arraysep
	& = &  j(\mu(b),n(d))
			\comp 
			\rho(b,\sigma(c)) \wideres j(g(b),\sigma(\tau))
			\arraysep
	& = &  j(\mu(b),n(d))
			\comp 
			\rho(b \res b,\sigma(c) \res \sigma(\tau)) 
			\arraysep
	& = &  j(\mu(b),n(d))
			\comp 
			\rho(b,n(d)) 
\end{array}
$ \\
The projection denotes the steps in $j(\mu(\pi),m(c)) \comp \rho(b,\sigma(c))$ that are not performed in $j(f(\pi),\sigma(\tau))$, namely the $\mu$ and $\rho$-steps, applied on the target of the latter proof term.
We remark that in the last step of this example, we obtain $b \res b = b$ by applying the fourth clause with $m = 0$.

\medskip
The projection operation is defined \emph{modulo} (the relation generated by) the equation \peqstruct. The following example shows why this is required.
\\
$
\begin{array}{rclcl}
\multicolumn{3}{l}{\rho(m(c),b) \wideres j(g(\sigma(c)) \comp g(n(\tau)),b) } 
	& = &  \rho(m(c),b) \wideres j(g(\sigma(c) \comp n(\tau)),b) \arraysep
	& = &  \rho(m(c) \wideres (\sigma(c) \comp n(\tau)), b \res b)
	& = &  \rho((m(c) \res \sigma(c)) \res n(\tau),b) \ \ 
	 = \ \   \rho(n(d),b)  
\end{array}
$ \\
Observe that $j(g(\sigma(c)) \comp g(n(\tau)),b)$ must be transformed into $j(g(\sigma(c) \comp n(\tau)),b)$ in order to apply the second clause in the definition of projection.

\medskip
In the following, we give a variant of the definition of the projection operation, aiming at two goals. 
First, to produce a more precise definition, making the use of structural equivalence explicit. 
Secondly, to obtain proof terms for projections involving \emph{infinite} reductions, at least in some cases.
For the first goal we establish the necessity, in some cases, to transform a proof term into a form that makes a fixed reduction prefix explicit.
This is the role of structural equivalence in the projection, as shown in the last developed example \wrt\ the fixed prefix $j(g(\Box),\Box)$.

\medskip
Let $C$ be a context having a finite number of holes, and $\psi$ a proof term.
We say that $C$ \emph{is a fixed prefix for} $\psi$, iff any of the following items apply: \\[2pt]
$
\begin{array}{@{\hspace*{3mm}\bullet\ \ }l}
C = \Box \arraysep
\psi = f(\psi_1, \ldots, \psi_m), C = f(C_1, \ldots, C_m), \textand C_i \textnormal{ is a fixed prefix for } \psi_i \textforall i \arraysep
\psi = \psi_1 \comp \psi_2 \textnormal{ or } \psi = \icomp \psi_i, \textand C \textnormal{ is a fixed prefix for } \psi_i \textforall i 
\end{array}
$

\smallskip\noindent
Observe that $C$ being a fixed prefix for $\psi$ implies that $C$ is composed by function (opposed to rule and dot) symbols only. $C$ being a fixed prefix is stable by \peqence.

Let $C$ be a context and $\psi$ a proof term, such that $C$ is a fixed prefix for $\psi$.
We define the \emph{explicit fixed-prefix form} of $\psi$ \wrt\ $C$, notation $\cfpceq{\psi}{C}$, as follows: \\[2pt]
$
\begin{array}{@{}r@{\ \ }c@{\ \ }l}
\cfpceq{\psi}{\Box} & \eqdef & \psi \arraysep
\cfpceq{f(\psi_1, \ldots, \psi_m)}{f(C_1, \ldots, C_m)} 
  & \eqdef & f(\cfpceq{\psi_1}{C_1}, \ldots, \cfpceq{\psi_m}{C_m}) \arraysep
\cfpceq{(\psi_1 \comp \psi_2)}{f(C_1, \ldots, C_m)}
  & \eqdef & f(\cfpceq{\psi_{11} \comp \psi_{21}}{C_1}, \ldots, \cfpceq{\psi_{1m} \comp \psi_{2m}}{C_m}) \arraysep & &  
	\textnormal{where } \cfpceq{\psi_i}{f^\Box} = f(\psi_{i1}, \ldots, \psi_{im}) \textnormal{ for } i = 1,2 \arraysep
\cfpceq{(\icomp \psi_i)}{f(C_1, \ldots, C_m)}
  & \eqdef & f(\cfpceq{\icomp \psi_{i1}}{C_1}, \ldots, \cfpceq{\icomp \psi_{im}}{C_m}) \arraysep & &  	
	\textnormal{where } \cfpceq{\psi_i}{f^\Box} = f(\psi_{i1}, \ldots, \psi_{im}) \textforall i < \omega	
\end{array}
$ \\[2pt]
In this definition, as well as in the sequel, $f^\Box$ denotes the context $f(\Box, \ldots, \Box)$. We use also $l^\Box$ and $h^\Box$, where $\mu : l \to h$.
Observe that 
$\cfpceq{j(g(\sigma(c)) \comp g(n(\tau)), b)}{j(g(\Box),\Box)}
  = j(\cfpceq{g(\sigma(c)) \comp g(n(\tau))}{g(\Box)}, \cfpceq{b}{\Box})
  = j(g(\cfpceq{\sigma(c) \comp n(\tau)}{\Box}), b)
  = j(g(\sigma(c) \comp n(\tau)), b)
	$, the form needed to compute the projection in the last given example.
	
\smallskip
We say that a proof term $\psi$ \emph{includes head steps}, if $\psi = \mu(\psi_1, \ldots, \psi_m)$, or either $\psi = \psi_1 \comp \psi_2$ or $\psi = \icomp \psi_i$, and some $\psi_n$ includes head steps.

\smallskip
Given two coinitial proof terms $\psi$ and $\phi$, we define the \emph{projection} of $\psi$ over $\phi$, notation $\psi \res \phi$, as the operation given by the following clauses, considered in order.

\begin{enumerate}[1.]
\item \label{it:proj-id} 
$src(\psi) \wideres \psi \eqdef tgt(\psi)
\qquad\qquad
\psi \wideres src(\psi) \eqdef \psi$
\vspace*{-2pt}

\item \label{it:proj-mumu} 
$\mu(\phi_1, \ldots, \phi_m) \wideres \mu(\psi_1, \ldots, \psi_m)
 \eqdef
h[\phi_1 \res \psi_1, \ldots, \phi_m \res \psi_m]$
\vspace*{-2pt}

\item \label{it:proj-mul} 
$\mu(\phi_1, \ldots, \phi_m) \wideres \psi
\eqdef
\mu(\phi_1 \res \psi_1, \ldots, \phi_m \res \psi_m)$
\qquad\qquad
if $l^\Box$ is a fixed prefix for $\psi$
\vspace*{-2pt}

\item \label{it:proj-lmu} 
$\phi \wideres \mu(\psi_1, \ldots, \psi_m)
\eqdef
h[\phi_1 \res \psi_1, \ldots, \phi_m \res \psi_m]$
\qquad\qquad
if $l^\Box$ is a fixed prefix for $\phi$

\item \label{it:proj-comp-left} 
$(\phi \comp \psi) \wideres \chi
\eqdef
\phi \res \chi \comp \psi \wideres (\chi \res \phi)$
\\
if $\phi \comp \psi$ includes head steps and $\chi = \mu(\chi_1, \ldots, \chi_m)$, $f(\chi_1, \ldots, \chi_m)$ or $\icomp \chi_i$

\item \label{it:proj-comp-right} 
$\chi \wideres (\phi \comp \psi)
\eqdef
(\chi \res \phi) \wideres \psi$
\qquad\qquad
if either $\phi \comp \psi$ or $\chi$ include head steps  

\item \label{it:proj-ff} 
$\phi \wideres \psi
\eqdef
f(\phi_1 \res \psi_1, \ldots, \phi_m \res \psi_m)$
\qquad
%
if $f^\Box$ is a fixed prefix for both $\phi$ and $\psi$
\end{enumerate} 

\noindent
where in clauses \ref{it:proj-mumu}, \ref{it:proj-mul} and \ref{it:proj-lmu}, $\mu: l \to h$;
and also
$\cfpceq{\psi}{l^\Box} = l[\psi_1, \ldots, \psi_m]$ in clause \ref{it:proj-mul}, 
$\cfpceq{\phi}{l^\Box} = l[\phi_1, \ldots, \phi_m]$ in clause \ref{it:proj-lmu}, and analogously for $\cfpceq{\phi}{f^\Box}$ and $\cfpceq{\psi}{f^\Box}$ in clause \ref{it:proj-ff}. We remark that clauses \ref{it:proj-comp-left} and \ref{it:proj-comp-right} apply to both binary and infinitary composition.


We add a few comments on this definition of projection.
First, when using the definition we will always consider proof terms modulo the relation generated by the equations \peqidleft, \peqidright\ and \peqassoc, the so-called \emph{reduction identities} in \cite{terese}.
Secondly, we assume 
that $\psi$ and $\phi$ are \emph{mutually orthogonal}, even if the underlying TRS is not. This implies in particular that if $\phi = \mu(\phi_1, \ldots, \phi_m)$ where $\mu: l \to h$, and $\psi$ does not include head steps, then $l^\Box$ is a fixed prefix for $\psi$.
Finally, we note that clause \ref{it:proj-id} is needed to avoid infinite iteration if $src(\psi)$ is an \emph{infinite} term. Otherwise, \eg\ to compute $f\om \res f\om$ clause \ref{it:proj-ff} would have to be applied ad infinitum.

We show some simple cases of projections involving infinite proof terms, using the rule $\mu: f(x) \to g(x)$.
Omitting parentheses for unary function symbols,
we have \eg\
$f \mu\om \res \mu f \mu f\om
	= g(\mu\om \res f \mu f\om)
	= g \mu (\mu\om \res \mu f\om)
	= g \mu g (\mu\om \res f\om)
	= g \mu g \mu\om$,   
applying clauses \ref{it:proj-lmu}, \ref{it:proj-mul}, \ref{it:proj-mumu} and \ref{it:proj-id} respecively.  We can also obtain the projection \emph{over} an infinite reduction:   
$\mu f \mu f\om \res f \mu\om 
  = \mu(f \mu f\om \res \mu\om)
	= \mu g (\mu f\om \res \mu\om)
	= \mu g g (f\om \res \mu\om)
	= \mu g\om$.%

Sequential reductions lead to more laborious projection computations: \\
$
\begin{array}{rcll}
\multicolumn{4}{l}{\icomp f g^i \mu f\om \wideres (\mu f\om \comp g f \mu f\om)} \arraysep
  & = & (\icomp f g^i \mu f\om \res \mu f\om) \wideres g f \mu f\om 
		& \textnormal{clause \ref{it:proj-comp-right}}	\arraysep
  & = & g (\icomp g^i \mu f\om \res f\om) \wideres g f \mu f\om 
		& \textnormal{clause \ref{it:proj-lmu}}	\arraysep
  & = & g (\icomp g^i \mu f\om) \wideres g f \mu f\om 
		& \textnormal{clause \ref{it:proj-id}}	\arraysep
  & = & g (\icomp g^i \mu f\om  \wideres f \mu f\om) 
		& \textnormal{clause \ref{it:proj-ff}}	\arraysep
  & = & g ((\mu f\om \res f \mu f\om) \comp (\icomp g^{i+1} \mu f\om \wideres (f \mu f\om \res \mu f\om)) ) 
		& \textnormal{clause \ref{it:proj-comp-left}}	\arraysep
  & = & g (\mu g f\om \comp (\icomp g^{i+1} \mu f\om \wideres g \mu f\om) ) \arraysep
  & = & g (\mu g f\om \comp g (\icomp g^{i} \mu f\om \wideres \mu f\om) ) 
		& \textnormal{clause \ref{it:proj-ff}}	\arraysep
  & = & g (\mu g f\om \comp 
		g (\mu f\om \res \mu f\om \comp (\icomp g^{i+1} \mu f\om \wideres (\mu f\om \res \mu f\om))) 
					)
		& \textnormal{clause \ref{it:proj-comp-left}}	\arraysep
  & = & g (\mu g f\om \comp 
		g (g f\om \comp (\icomp g^{i+1} \mu f\om \wideres g f\om)) 
					) \arraysep
  & = & g (\mu g f\om \comp 
		g (g f\om \comp \icomp g^{i+1} \mu f\om ) 
					) 
		& \textnormal{clause \ref{it:proj-id}}	\arraysep
  & = & g (\mu g f\om \comp g (\icomp g^{i+1} \mu f\om) ) 
		& \textnormal{reduction identities} 
\end{array}
$ 

\noindent
Note that the explicit fixed-prefix form of an infinite composition is used several times, namely, in the use of clause \ref{it:proj-lmu} and both uses of clause \ref{it:proj-ff}.

By relating the given examples, we observe that simultaneous and sequential descriptions of the same reduction lead to \peqent\ projections. 
In this case we have $g \mu g \mu\om \peq g (\mu g f\om \comp g (\icomp g^{i+1} \mu f\om) )$, as we prove in the following.
Note that for any $n < \omega$, using just \peqassoc\ we obtain 
$\icomp g^{i} \mu f\om \peqb \mu f\om \comp g \mu f\om \comp \ldots \comp g^n \mu f\om \comp \icomp g^{i+n+1} \mu f\om$. On the other hand, \peqoutin\ and \peqstruct\ yield
$\mu\om \peqb \mu f\om \comp g \mu\om \peqb \mu f\om \comp g (\mu f\om \comp g \mu\om) \peqb \mu f\om \comp g \mu f\om \comp g^2 \mu\om $, so that a simple iteration entails $\mu\om \peqb \mu f\om \comp g \mu f\om \comp \ldots \comp g^n \mu f\om \comp g^{n+1} \mu\om$. 
Hence \eqllim\ allows to assert $\icomp g^i \mu f\om \peq \mu\om$.
In turn, $g \mu g \mu\om \peq g (\mu g f\om \comp g^2 \mu\om)$ while
$g (\mu g f\om \comp g (\icomp g^{i+1} \mu f\om) ) \peq g (\mu g f\om \comp g^2 (\icomp g^{i} \mu f\om) )$, where \peqinfstruct\ is used for the latter assertion.
Hence, congruence allows to conclude.


Finally we remark that in the given examples, projections involving an infinite proof term are successively decomposed, until 
clause \ref{it:proj-id} can be used to obtain a final expression for the projection.
Limits are only indirectly involved, to compute source or target terms in the uses of that clause. 
In Section~\ref{sec:limitations} we discuss some examples of projections where 
limits should be used in a more  essential way.
\vspace{-2ex}

%% file: property.tex
The definition of infinitary projections given in Section~\ref{sec:projection} allows to study the statement $\psi \comp (\phi \res \psi) \ \peq \ \phi \comp (\psi \res \phi)$, that we described in Section~\ref{sec:projection-intro}.
Let us verify this property for the first example of Section~\ref{sec:projection}, where $\psi = f \mu\om$ and $\phi = \mu f \mu f\om$, and the 
projections are $\psi \res \phi = g \mu g \mu\om$ and $\phi \res \psi = \mu g\om$. 
We have \\[2pt]
$
\begin{array}{r@{\ \ }c@{\ \ }ll}
	\psi \comp \phi \res \psi \ = \ 
	f \mu\om \comp \mu g\om & \peq & 
			\mu\om \ \peq \ 
			\mu f\om \comp g \mu\om 
		& \peqinout, \peqoutin \arraysep
			& \peq & 
			\mu f\om \comp g (f \mu\om \comp \mu g\om)
		& \peqinout \arraysep
			& \peq & 
			\mu f\om \comp g (f \mu f\om \comp f g \mu\om \comp \mu g\om)
		& \peqoutin, \peqstruct \arraysep
			& \peq & 
			\mu f\om \comp g f \mu f\om \comp g (f g \mu\om \comp \mu g\om)
		& \peqstruct \arraysep
			& \peq & 
			\mu f \mu f\om \comp g \mu g \mu\om
		& \peqoutin, \peqinout \arraysep
		& = & \phi \comp \psi \res \phi
\end{array}
$

\medskip
This section is devoted to proving the above mentioned result in a very limited case; namely, when $\phi$ denotes a single step on the source term of $\psi$, that is actually included in $\psi$.
Moreover, we ask $\phi$ to denote a step \emph{easily extractable} from $\psi$.  The forthcoming statement covers \eg\ this case: $\psi = \mu\om$ or $\psi = \icomp g^i \mu f\om$, and $\phi = \mu f\om$.
The example just described is not comprised: $\phi$ denotes two (simultaneous) steps , and one of them (the outermost one) is not included in $\psi$.
\vspace{-2ex}

\subsection{Easily extractable steps}
Roughly speaking, a step included in a proof term $\psi$, that is, a rule symbol occurrence in $\psi$, is easily extractable if 
there are no other rule symbols in $\psi$ denoting activity performed before that step, that affect positions in its pattern (that is, in the left-hand side pattern that is replaced by that step) or above it.
\Eg, if $\mu: f(x) \to g(x)$, $\nu: g(x) \to k(x)$, and $\pi: a \to b$, then the only easily extractable step in $\mu(a) \comp \nu(\pi)$ is the $\mu$ occurrence%
, since it denotes a step that is performed before both the $\nu$- and the $\pi$-steps and affects the root position, the same as the $\nu$-step, and above that corresponding to the $\pi$-step.
On the other hand, both the $\mu$ and the $\pi$ occurrences are easily extractable in the equivalent $\mu(\pi) \comp \nu(b)$, since they are performed simultaneously.
We note that function symbols do not affect extractability, \eg\ all the rule symbol occurrences are easily extractable in $j(\mu(\pi),\nu(c))$.

Formally, we define the set of easily extractable rule symbol occurrences in a proof term $\psi$, notation $\esrs{\psi}$, as a set of \emph{pairs of positions}.
The left component is the \emph{contraction position}, \ie\ the position in $src(\psi)$ where the step can be applied. The right component is the \emph{position of the rule symbol occurrence} in the proof term.
\Eg, 
if $\psi = \mu(a) \comp \nu(\pi)$, the only element of $\esrs{\psi}$ is 
$\pair{\epsilon}{1}$: the $\mu$ occurrence at \underline{position $1$ in $\psi$} can be applied at \underline {position $\epsilon$ on $src(\psi) = f(a)$}. 

As the material of this section is deeply based on position analysis, we define an analogous to the fixed-prefix context property, given in terms of positions.
Let $\Spa$ be a set of positions and $\psi$ a proof term. We say that $\psi$ \emph{respects} $\Spa$ iff the latter is finite and prefix-closed, and any of the following applies 
\begin{itemize}
\item 
$\psi$ is an infinitary multistep, $\Spa \subseteq \Pos{\psi}$ and $\psi(p) \in \Sigma$ for all $p \in \Spa$.
\otherminitemsep
\item 
$\psi = \psi_1 \comp \psi_2$, or $\psi = \icomp \psi_i$, and all $\psi_i$ respect $\Spa$
\otherminitemsep
\item 
$\psi = f(\psi_1, \ldots, \psi_m)$ and either $\Spa = \emptyset$ or $\psi_i$ respects $\proj{\Spa}{i}$ for all $i$
\otherminitemsep
\item 
$\psi = \mu(\psi_1, \ldots, \psi_m)$ and $\Spa = \emptyset$
\end{itemize}
\textsep
where $\proj{\Spa}{i} \eqdef \set{p \setsthat ip \in \Spa}$, and $\psi$ is assumed not a multistep in the last two clauses.

It is easy to verify that: 
(1) for any proof term $\psi$ and context $C$, $C$ is a fixed prefix for $\psi$ iff $\psi$ respects the set of non-hole positions of $C$, (2)
if $\psi$ respects $\Spa$, then $src(\psi)(r) = tgt(\psi)(r)$ for all $r \in \Spa$, and (3) \peqence\ preserves the \emph{respects} property. \Confer\ \cite{phdcarlos}, Sec.~5.5.

We now give the formal definition of $\esrsfn$. 
\\[4pt]
$
\begin{array}{rcl}
\esrs{\mu(\psi_1, \ldots, \psi_m)} 
	& \eqdef 
	& \set{\pair{\epsilon}{\epsilon}} \ \cup \ 
	\set{\pair{r_1 r_2}{ip} \setsthat \pair{r_2}{p} \in \esrs{\psi_i} \land l(r_1) = x_i}
	\arraysep & & \textnormal{where } \mu: l \to h \arraysep
\esrs{f(\psi_1, \ldots, \psi_m)} 
	& \eqdef 
  & \underset{i}{\bigcup} \set{\pair{ir}{ip} \setsthat \pair{r}{p} \in \esrs{\psi_i}} \arraysep
\esrs{\psi_1 \comp \psi_2} 
	& \eqdef 
	& \set{\pair{r}{1p} \setsthat \pair{r}{p} \in \esrs{\psi_1}} 
	  \ \cup \ 
	\set{\pair{r}{2p} \setsthat \pair{r}{p} \in \esrs{\psi_2} 
  \arraysep & & \ \ \ \land \ \psi_1 \textnormal{ respects } \set{r' \setsthat r' < r} \cup (r \cdot \PPos{\psi_2(p)})} 
\arraysep
\esrs{\icomp \psi_i} 
	& \eqdef 
	& \set{\pair{r}{2^j1p} \setsthat \pair{r}{p} \in \esrs{\psi_j} 
  \arraysep & & \ \ \ \land \ \psi_i \textnormal{ respects } \set{r' \setsthat r' < r} \cup (r \cdot \PPos{\psi_j(p)})} \textforall i < j	
\end{array}
$ \\ [2pt]
where $\PPos{\mu} = \set{p \in l \setsthat l(p) \notin \thevar}$ and $\mu: l \to h$.

The set of easily extractable steps is restricted to keep the definition 
simple, avoiding non-trivial analysis of positions. 
\Eg\ in $\psi = \mu(a) \comp \nu(\pi)$, the $\pi$-step, while not included in $\esrs{\psi}$, could be performed on $src(\psi) = f(a)$.

We verify that all easily extractable rule symbol occurrences are indeed extractable (to the source of the proof term) rule symbol occurrences.
\begin{lemma}
\label{rsl:esrs-are-es}
Let $\psi$ be a proof term, and $\pair{r}{p} \in \esrs{\psi}$. Then 
$\psi(p)$ is a rule symbol, say $\psi(p) = \mu$, and $\subtat{src(\psi)}{r} = l[s_1, \ldots, s_k]$ where $\mu: l \to h$.
\end{lemma}
\begin{proof}
A simple induction on $\pair{\psi}{r}$ suffices\footnote{If $\psi = \mu(\psi_1, \ldots, \psi_m)$ or $\psi = f(\psi_1, \ldots, \psi_m)$, then $\psi_i$ is not smaller than $\psi$ \wrt\ its ordinal number layer if $\psi$ is a multistep; this is the reason to consider induction on pairs, adding $r$ as the second component.}. If $\psi = \mu(\psi_1, \ldots, \psi_m)$ and $r = r_1 r_2$, recall that $r_1 \neq \epsilon$, then we conclude by induction on $\pair{\psi_i}{r_2}$. If $\psi = f(\psi_1, \ldots, \psi_m)$, so that $r = i r'$, then induction on $\pair{\psi_i}{r'}$ suffices to conclude.

Assume that $\psi = \psi_1 \comp \psi_2$. 
If $p = 1p'$, implying $\pair{r}{p'} \in \esrs{\psi_1}$, then \ih\ applies to $\pair{\psi_1}{r}$. Recalling that $src(\psi) = src(\psi_1)$, the conclusions of the \ih\ suffice to conclude.
If $p = 2p'$,  implying $\pair{r}{p'} \in \esrs{\psi_2}$, then \ih\ on $\pair{\psi_2}{r}$ yields that $\psi(p) = \psi_2(p') = \mu$, and also that $\subtat{src(\psi_2)}{r} = \subtat{tgt(\psi_1)}{r} = l[t_1, \ldots, t_k]$ for some $t_1, \ldots, t_k$.
In turn, $\pair{r}{2p'} \in \esrs{\psi}$ implies that $\psi_1$ respects $\set{r' \setsthat r' < r} \cup (r \cdot \PPos{l})$, so that $\subtat{src(\psi)}{r} = \subtat{src(\psi_1)}{r} = l[s_1, \ldots, s_k]$.

If $\psi = \icomp \psi_i$, then an argument similar to that given for the previous case, where $p = 2^j 1 p'$ instead of $p = 2 p'$ suffices; an iteration over $\langle \psi_{j-1}, \ldots, \psi_0 \rangle$ is required to verify $\subtat{src(\psi_1)}{r} = l[s_1, \ldots, s_k]$.
\end{proof}

\vspace{-1ex}
The elements of $\esrs{\psi}$ correspond to the steps that can be extracted, \ie, applied to $src(\psi)$. The following definition formalises the notion of applying a rule symbol occurrence to a term.
Let $t$ be a term, $r$ a position, and $\mu: l \to h$ a rule, such that $\subtat{t}{r} = l[t_1, \ldots, t_m]$.
We define the \emph{insertion of $\mu$ into $t$ at position $r$} as follows: 
$\insertrs{t}{\mu}{r} \eqdef \repl{t}{\mu(t_1, \ldots, t_m)}{r}$.
\subsection{Basic properties}
\label{sec:basic-properties}

\vspace{-1.6ex}
In order to prove the main result of this section, some basic properties of explicit fixed-prefix forms, easily extractable steps, and projections are required.
We will state these auxiliary results, along with some description. 
Their proofs, straightforward once the proper induction principle is determined, are given in \cite{long-version}.

First, we verify that the explicit fixed-prefix forms of a proof term, as defined in Section~\ref{sec:projection}, are equivalent to that proof term.

\begin{lemma}
\label{rsl:cfpceq-well-defined}
Let $\psi$ be a proof term, and $C$ a context such that $C$ is a fixed prefix for $\psi$.
Then $\cfpceq{\psi}{C} = C[\psi_1, \ldots, \psi_m]$, and $\psi \peqb \cfpceq{\psi}{C}$. Moreover, these proof terms are \emph{structurally} equivalent, \ie, a \peqence\ derivation exists whose conclusion is $\cfpceq{\psi}{C} \peqb \psi$ and where neither \peqinout\ nor \peqoutin\ are used.
\end{lemma}

The next result states that easily extractable steps are compatible with explicit fixed-prefix forms, where \emph{the contraction position does not change}.
\Eg, consider $\psi = m(f(\pi)) \comp m(\mu(b))$, so that $\cfpceq{\psi}{m(\Box)} = m(f(\pi) \comp \mu(b))$. We have $\pair{1}{21} \in \esrs{\psi}$, denoting that the $\mu$-step at position 21 is easily extractable to the position 1; note that $src(\psi) = m(f(a))$. The element of $\esrs{\cfpceq{\psi}{m(\Box)}}$ for the same step is $\pair{1}{12}$. The position of the rule symbol changed, while the contraction position is the same.

\begin{lemma}
\label{rsl:esrs-trace-forward}
Let $\psi$ be a proof term and $f$ a function symbol, such that $f^\Box$ is a fixed prefix for $\psi$, and $r, p$ such that $\pair{r}{p} \in \esrs{\psi}$.
Then there exists $q \in \Pos{\cfpceq{\psi}{f^\Box}}$ such that 
$(\cfpceq{\psi}{f^\Box}) (q) = \psi(p)$ and
$\pair{r}{q} \in \esrs{\cfpceq{\psi}{f^\Box}}$.
\end{lemma}

The following lemmas state that projections behave as expected in two straightforward cases: the projection of one step over a reduction that respects the set of pattern positions of the left-hand side of the corresponding rule; and the projection of one step over a reduction that includes that step.

\begin{lemma}
\label{rsl:res-step-over-respects}
Let $\mu: l \to h$ be a rule, $r$ a position, and $\psi$ a proof term, such that $\psi$ respects $\set{r' \setsthat r' < r} \cup (r \cdot \PPos{\mu})$, $\subtat{src(\psi)}{r} = l[s_1, \ldots, s_m]$, and consequently $\subtat{tgt(\psi)}{r} = l[t_1, \ldots, t_m]$.
Then 
$\insertrs{src(\psi)}{\mu}{r} \wideres \psi = \insertrs{tgt(\psi)}{\mu}{r}$.
\end{lemma}

\begin{lemma}
\label{rsl:res-esrs-over-pterm}
Whenever $\pair{r}{p} \in \esrs{\psi}$, we have
$\insertrs{src(\psi)}{\psi(p)}{r} \wideres \psi = tgt(\psi)$.
\end{lemma}

\subsection{Main results}
\label{sec:property-main-results}
We prove that the projection behaves as expected, in the sense described at the beginning of this Section, \ie\ that $\psi \comp \phi \res \psi \peqb \phi \comp \psi \res \phi$, in two situations in which $\phi$ is a one-step. Firstly, if $\psi$ does not interfere with $\phi$, that is, if the activity described by $\psi$ neither overlaps nor embeds the step described by $\phi$. Secondly, if $\phi$ is an easily extractable step for $\psi$.

\begin{lemma}
\label{rsl:proj-over-irs-peqb-if-respects}
Let $\psi$ be a proof term, $\mu: l \to h$ a rule symbol, and $r$ a position, such that $\psi$ respects $\set{r' \setsthat r' < r} \cup (r \cdot \PPos{\mu})$ and $\subtat{src(\psi)}{r} = l[s_1, \ldots, s_m]$.
Then
$\insertrs{src(\psi)}{\mu}{r} \comp \psi \wideres \insertrs{src(\psi)}{\mu}{r}
  \peqb \psi \comp \insertrs{tgt(\psi)}{\mu}{r} = \psi \comp \insertrs{src(\psi)}{\mu}{r} \wideres \psi$; \confer\ Lemma~\ref{rsl:res-step-over-respects}.
\end{lemma}

\begin{proof} 
We give only a sketch here, the full details can be found in \cite{long-version}.
The statement can be proved by induction on $r$.

If $r = \epsilon$, then $l^\Box$ is a fixed prefix for $\psi$, so that we can consider $\cfpceq{\psi}{l^\Box} = l[\psi_1, \ldots, \psi_m]$; let $src(\psi_i) = s_i$ and $tgt(\psi_i) = t_i$ for all $i$.
It is easy to obtain $\insertrs{src(\psi)}{\mu}{r} = \mu(s_1, \ldots, s_m)$ and $\insertrs{tgt(\psi)}{\mu}{r} = \mu(t_1, \ldots, t_m)$.
Then 
$\insertrs{src(\psi)}{\mu}{r} \comp \psi \wideres \insertrs{src(\psi)}{\mu}{r}
	= \mu(s_1, \ldots, s_m) \comp h[\psi_1, \ldots, \psi_m] 
	\peqb \mu(\psi_1, \ldots, \psi_m)
	\peqb l[\psi_1, \ldots, \psi_m] \comp \mu(t_1, \ldots, t_m)$.
	
If $r = i r_1$, then $f^\Box$ is a fixed prefix for $\psi$ for some $f$, so that we have $\cfpceq{\psi}{f^\Box} = f(\psi_1, \ldots, \psi_m)$. If $src(\psi_i) = s_i$ for all $i$, then $\insertrs{src(\psi)}{\mu}{r} = f(s_1, \ldots, \insertrs{src(\psi_i)}{\mu}{r_1}, \ldots, s_m)$ and similarly for $\insertrs{tgt(\psi)}{\mu}{r}$. It turns out that \ih\ can be applied on $\psi_i$ which, along with structural equivalence, suffices to conclude.
\end{proof}

\begin{proposition}
\label{rsl:proj-over-irs-peqb}
Let $\psi$ be a proof term, and $\pair{r}{p} \in \esrs{\psi}$.
Then $\insertrs{src(\psi)}{\psi(p)}{r} \comp (\psi \wideres \insertrs{src(\psi)}{\psi(p)}{r}) \peqb \psi \peqb \psi \comp (\insertrs{src(\psi)}{\psi(p)}{r} \res \psi)$; \confer\ Lemma~\ref{rsl:res-esrs-over-pterm}.
\end{proposition}

\vspace*{-3mm}
\begin{proof}
We proceed by induction on $\pair{r}{p}$.

Assume that $r = p = \epsilon$, so that $\psi = \mu(\psi_1, \ldots, \psi_m)$ and 
$\insertrs{src(\psi)}{\psi(p)}{r} = \mu(src(\psi_1), \ldots, src(\psi_m))$. Say $\mu: l \to h$. We have
$\insertrs{src(\psi)}{\psi(p)}{r} \comp (\psi \wideres \insertrs{\psi}{\psi(p)}{r})
  = \mu(src(\psi_1), \ldots, src(\psi_m)) \comp h[\psi_1 \res src(\psi_1), \ldots, \psi_m \res src(\psi_m)]
	= \mu(src(\psi_1), \ldots, src(\psi_m)) \comp h[\psi_1, \ldots, \psi_m]
	\peqb \psi
$
applying \peqoutin\ in the last step. Note that clause \ref{it:proj-mumu} applies to $\psi \wideres \insertrs{src(\psi)}{\psi(p)}{r}$.

\smallskip
Assume that $\psi = \mu(\psi_1, \ldots, \psi_m)$ where $\mu: l \to h$, and $r \neq \epsilon$. 
In this case, $r = r_1 r_2$, $p = i p_2$, $l(r_1) = x_i$, and $\pair{r_2}{p_2} \in \esrs{\psi_i}$.
Observe that $\insertrs{src(\psi)}{\psi(p)}{r} = l[src(\psi_1), \ldots, \insertrs{src(\psi_i)}{\psi_i(p_2)}{r_2}, \ldots, src(\psi_m)]$.
\Ih\ on $\pair{r_2}{p_2}$ entails 
$\insertrs{src(\psi_i)}{\psi_i(p_2)}{r_2} \comp (\psi_i \wideres \insertrs{src(\psi_i)}{\psi_i(p_2)}{r_2}) \peqb \psi_i$, implying in particular that $tgt(\psi_i \wideres \insertrs{src(\psi_i)}{\psi_i(p_2)}{r_2}) = tgt(\psi_i)$.
We have \\
$
\begin{array}{@{\hspace*{6mm}}cl}
\multicolumn{2}{l}{\insertrs{src(\psi)}{\psi(p)}{r} \comp (\psi \wideres \insertrs{src(\psi)}{\psi(p)}{r})} \arraysep
= & l[src(\psi_1), \ldots, \insertrs{src(\psi_i)}{\psi_i(p_2)}{r_2}, \ldots, src(\psi_m)] \arraysep
     & \comp \ \mu(\psi_1, \ldots, \psi_i \wideres \insertrs{src(\psi_i)}{\psi_i(p_2)}{r_2}, \ldots, \psi_m) \arraysep
\peqb & l[src(\psi_1), \ldots, \insertrs{src(\psi_i)}{\psi_i(p_2)}{r_2}, \ldots, src(\psi_m)] \arraysep
     & \comp \ l[\psi_1, \ldots, \psi_i \wideres \insertrs{src(\psi_i)}{\psi_i(p_2)}{r_2}, \ldots, \psi_m] \arraysep
     & \comp \ \mu(tgt(\psi), \ldots, tgt(\psi_i), \ldots, tgt(\psi_m)) \arraysep
\peqb & l[\psi_1, \ldots, \insertrs{src(\psi_i)}{\psi_i(p_2)}{r_2} \comp (\psi_i \wideres \insertrs{src(\psi_i)}{\psi_i(p_2)}{r_2}), \ldots, \psi_m] \arraysep
      & \comp \ \mu(tgt(\psi), \ldots, tgt(\psi_i), \ldots, tgt(\psi_m)) \arraysep
\peqb & l[\psi_1, \ldots, \psi_i, \ldots, \psi_m] \comp \mu(tgt(\psi), \ldots, tgt(\psi_i), \ldots, tgt(\psi_m)) \ \ 
\peqb \ \  \psi
\end{array}
$ \\
by: definition of projection, where clause \ref{it:proj-mul} applies to $\psi \wideres \insertrs{src(\psi)}{\psi(p)}{r}$ and clause \ref{it:proj-id} to assert $\psi_j \res src(\psi_j) = \psi_j$ if $j \neq i$; \peqinout; structural equivalence including \peqstruct\ and \peqidleft; \ih\ as described above; and finally \peqinout.

\smallskip
Assume $\psi = \psi_1 \comp \psi_2$ is a binary composition that includes head steps. In this case $p = j p_1$, $src(\psi) = src(\psi_1)$, $\psi(p) = \psi_j(p_1)$, and $\pair{r}{p_1} \in \esrs{\psi_j}$. Clause \ref{it:proj-comp-left} applies to $\psi \wideres \insertrs{src(\psi)}{\psi(p)}{r}$, so that 
$\insertrs{src(\psi)}{\psi(p)}{r} \comp (\psi \wideres \insertrs{src(\psi)}{\psi(p)}{r}) = \insertrs{src(\psi)}{\psi(p)}{r} \comp (\psi_1 \wideres \insertrs{src(\psi)}{\psi(p)}{r}) \comp (\psi_2 \wideres (\insertrs{src(\psi)}{\psi(p)}{r} \res \psi_1))$.
\begin{itemize}
\item 
If $j = 1$, then \ih\ on $\pair{r}{p_1}$ yields $\insertrs{src(\psi)}{\psi(p)}{r} \comp (\psi_1 \wideres \insertrs{src(\psi)}{\psi(p)}{r}) \peqb \psi_1$, and Lemma~\ref{rsl:res-esrs-over-pterm} implies $\insertrs{src(\psi)}{\psi(p)}{r} \res \psi_1 = tgt(\psi_1) = src(\psi_2)$.
Consequently, 
$\insertrs{src(\psi)}{\psi(p)}{r} \comp (\psi \wideres \insertrs{src(\psi)}{\psi(p)}{r}) \peqb \psi_1 \comp (\psi_2 \res src(\psi_2)) = \psi$.
\item 
If $j = 2$, recall that $\psi_1$ respects $\set{r' \setsthat r' < r} \cup (r \cdot \PPos{\psi(p)}$. Moreover, Lemma~\ref{rsl:esrs-are-es} implies that $\subtat{src(\psi)}{r} = l[s_1, \ldots, s_m]$. Then Lemma~\ref{rsl:res-step-over-respects} implies $\insertrs{src(\psi)}{\psi(p)}{r} \res \psi_1 = \insertrs{src(\psi_2)}{\psi(p)}{r}$, and Lemma~\ref{rsl:proj-over-irs-peqb-if-respects} entails 
$\insertrs{src(\psi)}{\psi(p)}{r} \comp (\psi_1 \wideres \insertrs{src(\psi)}{\psi(p)}{r}) \peqb \psi_1 \comp \insertrs{src(\psi_2)}{\psi(p)}{r}$.
Consequently, application of clause \ref{it:proj-comp-left} yields
$\insertrs{src(\psi)}{\psi(p)}{r} \comp (\psi \wideres \insertrs{src(\psi)}{\psi(p)}{r}) \peqb \psi_1 \comp \insertrs{src(\psi_2)}{\psi(p)}{r} \comp \psi_2 \wideres \insertrs{src(\psi_2)}{\psi(p)}{r})$.
In turn, \ih\ on $\pair{r}{p_1}$ yields $\insertrs{src(\psi_2)}{\psi(p)}{r} \comp (\psi_2 \wideres \insertrs{src(\psi_2)}{\psi(p)}{r} \peqb \psi_2$; recall that $\psi(p) = \psi_2(p_1)$. Hence we conclude.
\end{itemize}

Assume $\psi = \icomp \psi_i$ and $\psi$ includes head steps. Then $p = 2^j 1 p_1$, where $\pair{r}{p_1} \in \esrs{\psi_j}$, $\psi_i$ respects $\set{r' \setsthat r' < r} \cup (r \comp \PPos{\psi(p)}$, and $\psi(p) = \psi_j(p_1)$.
Lemma~\ref{rsl:esrs-are-es} implies $\subtat{src(\psi)}{r} = \subtat{src(\psi_0)}{r} = l[s_{01}, \ldots, s_{0m}]$.
Clause \ref{it:proj-comp-left} yields $\psi \wideres \insertrs{src(\psi)}{\psi(p)}{r} = (\psi_0 \res \insertrs{src(\psi)}{\psi(p)}{r}) \comp \icomp \psi_{i+1} \wideres (\insertrs{src(\psi)}{\psi(p)}{r} \res \psi_0)$. In turn, Lemma~\ref{rsl:proj-over-irs-peqb-if-respects} implies 
$\insertrs{src(\psi)}{\psi(p)}{r} \comp (\psi_0 \res \insertrs{src(\psi)}{\psi(p)}{r}) \peqb \psi_0 \comp \insertrs{src(\psi_1)}{\psi(p)}{r}$, and Lemma~\ref{rsl:res-step-over-respects} entails $\insertrs{src(\psi)}{\psi(p)}{r} \res \psi_0 = \insertrs{src(\psi_1)}{\psi(p)}{r}$.
Therefore, $\insertrs{src(\psi)}{\psi(p)}{r} \comp (\psi \res \insertrs{src(\psi)}{\psi(p)}{r}) \peqb \psi_0 \comp \insertrs{src(\psi_1)}{\psi(p)}{r} \comp \icomp \psi_{i+1} \res \insertrs{src(\psi_1)}{\psi(p)}{r}$.
This argument can be iterated for all $n < j$; observe $\pair{r}{2^{j-n}1p_1} \in \esrs{\icomp \psi_{i+n}}$ and $\psi(p) = \icomp \psi_{i+n}(2^{j-n} 1 p_1)$.
We obtain
$\insertrs{src(\psi)}{\psi(p)}{r} \comp (\psi \res \insertrs{src(\psi)}{\psi(p)}{r}) \peqb \psi_0 \comp \ldots \comp \psi_{j-1} \comp \insertrs{src(\psi_j)}{\psi(p)}{r} \comp \icomp \psi_{i+j} \res \insertrs{src(\psi_j)}{\psi(p)}{r}$.
\Ih\ applies on $\pair{r}{p_1}$, allowing to assert $\insertrs{src(\psi_j)}{\psi(p)}{r} \comp \icomp \psi_{i+j} \res \insertrs{src(\psi_j)}{\psi(p)}{r} \peqb \icomp \psi_{i+j}$; recall $\psi(p) = \psi_j(p_1)$. This suffices to conclude.

\smallskip
Assume that $f^\Box$ is a fixed prefix for $\psi$, where $src(\psi) = f(s_1, \ldots, s_m)$. It is easy to obtain $r \neq \epsilon$, that is, $r = i r_1$. Say $\cfpceq{\psi}{f^\Box} = f(\psi_1, \ldots, \psi_m)$. Then $src(\psi) = src(\cfpceq{\psi}{f^\Box})$ implies $src(\psi_i) = s_i$. Moreover, Lemma~\ref{rsl:esrs-trace-forward} implies $\pair{r}{q} \in \esrs{\cfpceq{\psi}{f^\Box}}$ for some $q$ such that $\psi(p) = \cfpceq{\psi}{f^\Box}(q)$. In turn, this implies $q = i q_1$ and $\psi(p) = \psi_i(q_1)$. Therefore, 
$\insertrs{src(\psi)}{\psi(p)}{r}
  = \insertrs{src(f(\psi_1, \ldots, \psi_m)}{\psi_i(q_1)}	{i r_1} $ \newlineifdraft $
	= f(src(\psi_1), \ldots, \insertrs{src(\psi_i)}{\psi_i(q_1)}{r_1}, \ldots, src(\psi_m))$.
Clause \ref{it:proj-ff} applies to $\psi \res \insertrs{src(\psi)}{\psi(p)}{r}$, so that \\
$
\begin{array}{@{\hspace*{6mm}}cl}
\multicolumn{2}{l}{\insertrs{src(\psi)}{\psi(p)}{r} \comp (\psi \wideres \insertrs{src(\psi)}{\psi(p)}{r})} \arraysep
= & f(src(\psi_1), \ldots, \insertrs{src(\psi_i)}{\psi_i(q_1)}{r_1}, \ldots, src(\psi_m)) \arraysep
      & \comp \ f(\psi_1 , \ldots, \psi_i \wideres \insertrs{src(\psi_i)}{\psi_i(q_1)}{r_1}, \ldots, \psi_m ) \arraysep
\peqb & f(\psi_1, \ldots, \insertrs{src(\psi_i)}{\psi_i(q_1)}{r_1} \comp \psi_i \wideres \insertrs{src(\psi_i)}{\psi_i(q_1)}{r_1}, \ldots, \psi_m) \arraysep
\peqb & f(\psi_1, \ldots, \psi_i, \ldots, \psi_m) \ = \ \cfpceq{\psi}{f^\Box} \ \peqb \ \psi
\end{array}
$ \\
by definition of projection where clause \ref{it:proj-id} yields $\psi_j \res s_j =  \psi_j$ if $j \neq i$; structural equivalence; and \ih\ on $\pair{r_1}{q_1}$ along with Lemma~\ref{rsl:cfpceq-well-defined}.
\end{proof}

%% file: limitations.tex
As shown by the discussion at the beginning of Section~\ref{sec:property}, the definitions given in Section~\ref{sec:projection} allow to obtain proper projections for cases beyond the scope of Lemma~\ref{rsl:proj-over-irs-peqb-if-respects} and Prop.~\ref{rsl:proj-over-irs-peqb}.
However, this is not always the case, even for projections involving an infinite and a finite reduction.

As an example, consider the rules $\rho: gx \to fgx$, $\pi : a \to b$, and let $\psi = \icomp f^i \rho a$, $\phi = g \pi$.
We claim that according to the intuitive notion of projection, the result of $\psi \res \phi$ should be $\icomp f^{i} \rho b$, that is the same reduction denoted by $\psi$, applied to the target of $\phi$, namely $g(b)$. \Wrt\ $\phi \res \psi$, we note that the $\pi$ step denoted by $\phi$ vanishes in $tgt(\psi) = f\om$, while it can be performed on each partial target $f^n g a$. This phenomenon is referred to as \emph{infinitary erasure} in \cite{nosotros-rta14}. Accordingly, we could expect the result of $\phi \res \psi$ to be $f\om$.

We have 
$\psi \res \phi 
  = (\rho a \res g \pi) \comp \icomp f^{i+1} \rho a \wideres (g \pi \res \rho a)
  = \rho b \comp \icomp f^{i+1} \rho a \wideres fg \pi
  = \rho b \comp f(\icomp f^{i} \rho a \wideres g \pi)
  = \rho b \comp f(\psi \res \phi)
$, where 
the first and third equalities are justified by clauses \ref{it:proj-comp-left} and \ref{it:proj-ff} resp.,
and the last one just considers the definitions of $\psi$ and $\phi$.
Successive iterations yield $\rho b \comp f(\rho b \comp f(\psi \res \phi))$, $\rho b \comp f(\rho b \comp f(\rho b \comp f(\psi \res \phi)))$, etc., \ie, we obtain always expressions including an occurrence of the projection operator. On the other hand, 
$\phi \res \psi 
  = (g \pi \res \rho a) \wideres \icomp f^{i+1} \rho a
  = fg \pi \wideres \icomp f^{i+1} \rho a
  = f(g \pi \wideres \icomp f^{i} \rho a)
  = f(\phi \res \psi)
  = f^2(\phi \res \psi)
  \ldots$, 
where clauses \ref{it:proj-comp-right} and \ref{it:proj-ff} are used in the first and third equalities resp.. As in the previous case, the successive expressions obtained always include an occurrence of the projection operator.
This differs from the behaviour of the examples in Section~\ref{sec:projection}, where a final (\ie\ without occurrences of the projection operator) expression is obtained.

Observe that in both cases, the partial results approximate the expected final results. A similar phenomenon occurs when applying our definition to obtain the projection of an infinite composition over another one.
These observations suggest the need of incorporating the notion of limit in the proposed definition of projection,  in order to cover the cases not currently considered.

%% file: conclusions.tex
In this article, we describe our work-in-progress about a possible characterisation, based on proof terms, of the projection of one reduction over another for infinitary, left-linear, first-order rewriting. We introduce this characterisation, show that it conveys the expected results in several cases, and prove a partial confluence property. We also discuss some limitations of the current form of the characterisation.

Two obvious further directions of work are: to extend the proposed definition, in order to comprise all projections of an infinitary reduction over another one, and to extend the soundness property expressed in Prop.~\ref{rsl:proj-over-irs-peqb} to all projections. 
Additionally, it would be interesting to further delimit the scope of the current version, that is, to understand in which cases the development of a projection can be performed without explicit use of the notion of limit.

%% file: appendixes.tex
\appendix

\section{Proofs of auxiliary lemmas}
\label{sec:additional-proofs}

\input{additional-proofs}

\section{About the scope of clauses \ref{it:proj-comp-left} and \ref{it:proj-comp-right}}

\input{proj-comp-examples}

%% file: additional-proofs.tex
In this section, we include the proofs of the results stated in Section~\ref{sec:basic-properties}, and the complete proof of Lemma~\ref{rsl:proj-over-irs-peqb-if-respects}, whose statement is given in Section~\ref{sec:property-main-results} along with a proof sketch.

\medskip\noindent
\textbf{Lemma~\ref{rsl:cfpceq-well-defined}.} We recall the statement: 
\begin{quote}
Let $\psi$ be a proof term, and $C$ a context such that $C$ is a fixed prefix for $\psi$.
Then $\cfpceq{\psi}{C} = C[\psi_1, \ldots, \psi_m]$, and $\psi \peqb \cfpceq{\psi}{C}$. Moreover, these proof terms are \emph{structurally} equivalent, \ie, a \peqence\ derivation exists whose conclusion is $\cfpceq{\psi}{C} \peqb \psi$ and where neither \peqinout\ nor \peqoutin\ are used.
\end{quote}
\begin{proof}
We proceed by induction on $\pair{C}{\psi}$.
If $C = \emptyset$ then the result holds immediately.
Therefore, we assume $C = f(C_1, \ldots, C_m)$ in the sequel.

Assume that $\psi = f(\psi_1, \ldots, \psi_m)$. In this case, $\cfpceq{\psi}{C} = f(\cfpceq{\psi_1}{C_1}, \ldots, \cfpceq{\psi_m}{C_m})$.
For each $i$, we can apply \ih\ on $\pair{C_i}{\psi_i}$.
Therefore, we obtain $\cfpceq{\psi}{C} = C[\psi_1, \ldots, \psi_m]$ immediately, and $\cfpceq{\psi}{C} \peqb \psi$ just by congruence.

Assume that $\psi = \psi_1 \comp \psi_2$.
Then for $i = 1,2$, we can apply \ih\ on $\pair{f^\Box}{\psi_i}$; note that $f^\Box$ coincides, or is simpler than, $C$.
We obtain that $\cfpceq{\psi_i}{f^\Box} = f(\psi_{i1}, \ldots, \psi_{im}) \peqb \psi_i$.
In turn, $\psi_1 \peqb f(\psi_{11}, \ldots, \psi_{1m})$ and $\psi_2 \peqb f(\psi_{21}, \ldots, \psi_{2m})$ imply 
$\psi \peqb f(\psi_{11} \comp \psi_{21}, \ldots, \psi_{1m} \comp \psi_{2m})$, using \peqstruct\ and congruence.
Therefore, $C$ is a fixed prefix for $f(\psi_{11} \comp \psi_{21}, \ldots, \psi_{1m} \comp \psi_{2m})$ (recall that being a fixed prefix is stable by \peqence).
In turn, \ih\ applies to $\pair{C_j}{\psi_{1j} \comp \psi_{2j}}$ for each $j$, so that $\cfpceq{\psi}{C} = C[\psi_1, \ldots, \psi_m]$ follows immediately.
Moreover, \ih\ yields $\cfpceq{\psi_{1i} \comp \psi_{2i}}{C_i} \peqb \psi_{1i} \comp \psi_{2i}$, so that we obtain
$\psi 
  \peqb f(\psi_{11} \comp \psi_{21}, \ldots, \psi_{1m} \comp \psi_{2m})
	\peqb f(\cfpceq{\psi_{11} \comp \psi_{21}}{C_1}, \ldots, \cfpceq{\psi_{1m} \comp \psi_{2m}}{C_m})
	= \cfpceq{\psi}{C}$ by congruence.
	
Assume that $\psi = \icomp \psi_i$. 
As in the previous case, \ih\ can be applied on $\pair{f^\Box}{\psi_i}$, now for each $i < \omega$. From $\psi_i \peqb f(\psi_{i1}, \ldots, \psi_{im})$ for each $i$, we obtain $\psi \peqb f(\icomp \psi_{i1}, \ldots, \icomp \psi_{im})$ by means of \eqlinfcomp, \peqinfstruct\ and transitivity, so that $C$ is a fixed prefix for the last proof term.
\Ih\ can be applied on $\pair{C_j}{\icomp \psi_{ji}}$ for each $j$; hence, the argument given for binary composition is valid in this case.
\end{proof}

\medskip\noindent
\textbf{Lemma~\ref{rsl:esrs-trace-forward}.} We recall the statement: 
\begin{quote}
Let $\psi$ be a proof term and $f$ a function symbol, such that $f^\Box$ is a fixed prefix for $\psi$, and $r, p$ such that $\pair{r}{p} \in \esrs{\psi}$.
Then there exists $q \in \Pos{\cfpceq{\psi}{f^\Box}}$ such that 
$(\cfpceq{\psi}{f^\Box}) (q) = \psi(p)$ and
$\pair{r}{q} \in \esrs{\cfpceq{\psi}{f^\Box}}$.
\end{quote}
\begin{proof}
We proceed by induction on $\psi$. Observe that $\psi = \mu(\psi_1, \ldots, \psi_m)$ would contradict 
$f^\Box$ to be a fixed prefix for $\psi$.
If $\psi = f(\psi_1, \ldots, \psi_m)$, then $\cfpceq{\psi}{f^\Box} = \psi$, so that it suffices to take $q = p$.

Assume $\psi = \psi_1 \comp \psi_2$. In this case, 
$\cfpceq{\psi}{f^\Box} = f(\psi_{11} \comp \psi_{21}, \ldots, \psi_{1m} \comp \psi_{2m})$, where $\cfpceq{\psi_i}{f^\Box} = f(\psi_{i1}, \ldots, \psi_{im})$ for $i = 1,2$; $p = j p_1$ where either $j = 1$ or $j = 2$; and $\pair{r}{p_1} \in \esrs{\psi_j}$.
\Ih\ on $\psi_j$ entails the existence of some $q_1$ that verifies $f(\psi_{j1}, \ldots, \psi_{jm})(q_1) = \psi_j(p_1) = \psi(p)$, and $\pair{r}{q_1} \in \esrs{f(\psi_{j1}, \ldots, \psi_{jm})}$.
The latter assertion implies 
the existence of $k$ such that $1 \leq k \leq m$, 
$r = k r_2$, $q_1 = k q_2$, and $\pair{r_2}{q_2} \in \esrs{\psi_{jk}}$. In turn, $q_1 = k q_2$ implies that $\psi(p) = f(\psi_{j1}, \ldots, \psi_{jm})(k q_2) = \psi_{jk}(q_2)$. 
If $j = 1$, then it is immediate that $\pair{r_2}{j q_2} \in \esrs{\psi_{1k} \comp \psi_{2k}}$.
If $j = 2$, then $\pair{r}{2p_1} \in \esrs{\psi}$ implies that $\psi_1$ respects $\set{r' \setsthat r' < k r_2} \cup (k r_2 \cdot \PPos{\psi(p)})$, so that $f(\psi_{11}, \ldots, \psi_{1m})$ does.
Therefore, $\psi_{1k}$ respects $\set{r' \setsthat r' < r_2} \cup (r_2 \cdot \PPos{\psi_{2k}(q_2)})$.
Hence, we have again $\pair{r_2}{j q_2} \in \esrs{\psi_{1k} \comp \psi_{2k}}$.
We take $q = k j q_2$. A straightforward analysis suffices to conclude.

A similar analysis of that given for $j = 2$ applies if $\psi = \icomp \psi_i$, considering that $\cfpceq{\psi}{f^\Box} = f(\icomp \psi_{i1}, \ldots, \icomp \psi_{im})$, where $\cfpceq{\psi_i}{f^\Box} = f(\psi_{i1}, \ldots, \psi_{im})$ for all $i < \omega$, $p = 2^j 1 p_1$, where $\pair{r}{p_1} \in \esrs{\psi_j}$, and $\psi_i$ respects $\set{r' \setsthat r' < r} \cup (r \cdot \PPos{\psi_j(p_1)})$ for all $i < j$.
\end{proof}

\medskip\noindent
\textbf{Lemma~\ref{rsl:res-step-over-respects}.} We recall the statement: 
\begin{quote}
Let $\mu: l \to h$ be a rule, $r$ a position, and $\psi$ a proof term, such that $\psi$ respects $\set{r' \setsthat r' < r} \cup (r \cdot \PPos{\mu})$, $\subtat{src(\psi)}{r} = l[s_1, \ldots, s_m]$, and consequently $\subtat{tgt(\psi)}{r} = l[t_1, \ldots, t_m]$.
Then 
$\insertrs{src(\psi)}{\mu}{r} \wideres \psi = \insertrs{tgt(\psi)}{\mu}{r}$.
\end{quote}
\begin{proof}
We proceed by induction on $r$.

Assume $r = \epsilon$, implying that $\psi$ respects $\PPos{\mu}$. 
Lemma~\ref{rsl:cfpceq-well-defined} implies that $\cfpceq{\psi}{l^\Box} = l[\psi_1, \ldots, \psi_m] \peqb \psi$, so that $src(\psi) = l[src(\psi_1), \ldots, src(\psi_m)]$ and analogously for target.
We have \\
$
\begin{array}{@{}ll}
\insertrs{src(\psi)}{\mu}{r} \wideres \psi 
  \ = \  \mu(src(\psi_1), \ldots, src(\psi_m)) \wideres \psi
\arraysep \quad
  \ = \ \mu(src(\psi_1) \res \psi_1, \ldots, src(\psi_m) \res \psi_m) \arraysep \quad
  \ = \  \mu(tgt(\psi_1), \ldots, tgt(\psi_m)) 
  \ = \  \insertrs{tgt(\psi)}{\mu}{r}
\end{array}
$ \\
by clauses \ref{it:proj-mul} and \ref{it:proj-id}.

Assume $r = i r_1$, say $src(\psi)(\epsilon) = f$. 
Observe that $f^\Box$ is a fixed prefix for $\psi$, so that $\cfpceq{\psi}{f^\Box} = f(\psi_1, \ldots, \psi_m)$. 
Lemma~\ref{rsl:cfpceq-well-defined} implies $\psi \peqb f(\psi_1, \ldots, \psi_m)$, so that their source and target terms coincide. 
Observe that $f(\psi_1, \ldots, \psi_m)$ respects $\set{r' > r} \cup (r \cdot \PPos{\mu})$ since $\psi$ does, and therefore, that $\psi_i$ respects $\set{r' > r_1} \cup (r_1 \cdot \PPos{\mu})$. Moreover, 
$\subtat{src(\psi)}{r} = \subtat{src(f(\psi_1, \ldots, \psi_m))}{r} = \subtat{src(\psi_i)}{r_1}$, and analogously for the targets.
Consequently, we can apply \ih\ on $r_1$, obtaining that $\insertrs{src(\psi_i)}{\mu}{r_1} \wideres \psi_i = \insertrs{tgt(\psi_i)}{\mu}{r_1}$.
We have \\
$
\begin{array}{@{}ll}
\insertrs{src(\psi)}{\mu}{r} \wideres \psi \arraysep \quad
  = f(src(\psi_1), \ldots, \insertrs{src(\psi_i)}{\mu}{r_1}, \ldots, src(\psi_m)) \wideres f(\psi_1, \ldots, \psi_m) \arraysep \quad
	= f(src(\psi_1) \res \psi_1, \ldots,  \insertrs{src(\psi_i)}{\mu}{r_1} \res \psi_i, \ldots, src(\psi_1) \res \psi_1) \arraysep \quad
	= f(tgt(\psi_1), \ldots, \insertrs{tgt(\psi_i)}{\mu}{r_1}, \ldots, tgt(\psi_m)) \arraysep \quad
	= \insertrs{tgt(f(\psi_1, \ldots, \psi_m))}{\mu}{r} 
	\ = \  \insertrs{tgt(\psi)}{\mu}{r}
\end{array}
$ \\
by clauses \ref{it:proj-ff} and \ref{it:proj-id}. Thus we conclude.
\end{proof}

\medskip\noindent
\textbf{Lemma~\ref{rsl:res-esrs-over-pterm}.} We recall the statement: 
\begin{quote}
Whenever $\pair{r}{p} \in \esrs{\psi}$, we have
$\insertrs{src(\psi)}{\psi(p)}{r} \wideres \psi = tgt(\psi)$.
\end{quote}
\begin{proof}
We proceed by induction on $\pair{r}{p}$.

Assume $r = \epsilon$ and $\psi = \mu(\psi_1, \ldots, \psi_m)$, let us say $\mu: l \to h$.
In this case $p = \epsilon$, so that $\insertrs{src(\psi)}{\psi(p)}{r}) = \mu(src(\psi_1), \ldots, src(\psi_m))$. We have 
$\insertrs{src(\psi)}{\psi(p)}{r} \wideres \psi 
	= h[src(\psi_1) \res \psi_1, \ldots, src(\psi_m) \res \psi_m]
	= tgt(\psi)
$.

Assume that $r \neq \epsilon$ and $\psi = \mu(\psi_1, \ldots, \psi_m)$, say $\mu: l \to h$. In this case $r = r_1 r_2$ and $p = i p_1$, where $l(r_1) = x_i$ and $\pair{r_2}{p_1} \in \esrs{\psi_i}$; recall $\psi(p) = \psi_i(p_1)$.
Observe that $src(\psi) = l[src(\psi_1), \ldots, src(\psi_m)]$. Then
$\insertrs{src(\psi)}{\psi(p)}{r} \res \psi
  = l[src(\psi_1), \ldots, \insertrs{src(\psi_i)}{\psi_i(p_1)}{r_2}, \ldots, src(\psi_m)] \wideres \mu(\psi_1, \ldots, \psi_m)
	= $ \\ $h[src(\psi_1) \res \psi_1, \ldots, \insertrs{src(\psi_i)}{\psi_i(p_1)}{r_2}\res \psi_i, \ldots, src(\psi_m) \res \psi_m]$, note that clause \ref{it:proj-lmu} applies. 
\Ih\ on $\pair{r_2}{p_1}$ entails 
$\insertrs{src(\psi_i)}{\psi_i(p_1)}{r_2} \res \psi_i = tgt(\psi_i)$. On the other hand, if $j \neq i$ then $src(\psi_j) \res \psi_j = tgt(\psi_j)$. Hence we conclude.

Assume that $\psi = \psi_1 \comp \psi_2$, an either binary or infinite composition, $\psi$ includes head steps, and an arbitrary $r$. 
In this case, $p = j p_1$ and $\pair{r}{p_1} \in \esrs{\psi_i}$. 
Recall that $\psi(p) = \psi_j(p_1)$, $src(\psi) = src(\psi_1)$, $tgt(\psi_1) = src(\psi_2)$ and $tgt(\psi) = tgt(\psi_2)$. Clause \ref{it:proj-comp-right} yields 
$\insertrs{src(\psi)}{\psi(p)}{r} \wideres \psi 
  = (\insertrs{src(\psi_1)}{\psi(p)}{r} \res \psi_1) \wideres \psi_2$.
If $p = 1$, then \ih\ on $\pair{r}{p_1}$ entails 
$(\insertrs{src(\psi_1)}{\psi(p)}{r} \res \psi_1) \wideres \psi_2
  = src(\psi_2) \res \psi_2$, so that clause \ref{it:proj-id} allows to conclude.
If $p = 2$, then $\psi_1$ respects $\set{r' \setsthat r' < r} \comp \PPos{\psi(p)}$. Moreover, Lemma~\ref{rsl:esrs-are-es} implies $\subtat{src(\psi)}{r} = l[s_1, \ldots, s_m]$, where $\psi(r) : l \to r$. 
Then Lemma~\ref{rsl:res-step-over-respects} applies, yielding 
$(\insertrs{src(\psi_1)}{\psi(p)}{r} \res \psi_1) \wideres \psi_2
  = \insertrs{src(\psi_2)}{\psi(p)}{r} \res \psi_2$. Hence, \ih\ on $\pair{r}{p_1}$ suffices to conclude.

Assume that $f^\Box$ is a fixed prefix for $\psi$. In this case $src(\psi) = f(s_1, \ldots, s_m)$ and $r = i r_1$, so that
$\insertrs{src(\psi)}{\psi(p)}{r}) = f(s_1, \ldots, \insertrs{s_i}{\psi(p)}{r_1}, \ldots, s_m)$.
Let us say $\cfpceq{\psi}{f^\Box} = f(\psi_1, \ldots, \psi_i, \ldots, \psi_m)$; observe $src(\psi_i) = s_i$ for all $i$.
Lemma~\ref{rsl:esrs-trace-forward} implies the existence of some $q$ such that $\psi(p) = f(\psi_1, \ldots, \psi_i, \ldots, \psi_m) (q)$ and $\pair{r}{q} \in \esrs{f(\psi_1, \ldots, \psi_i, \ldots, \psi_m)}$.
In turn, the latter assertion entails that $q = i q_1$ and $\pair{r_1}{q_1} \in \esrs{\psi_i}$; observe that $\psi(p) = \psi_i(q_1)$. \Ih\ on $\pair{r_1}{q_1}$ allows to assert that $\insertrs{s_i}{\psi_i(q_1)}{r_1} \res \psi_i = tgt(\psi_i)$. Then clause~\ref{it:proj-ff} yields \\[2pt]
$\insertrs{src(\psi)}{\psi(p)}{r} \wideres \psi
  = f(s_1 \res \psi_1, \ldots, \insertrs{s_i}{\psi_i(q_1)}{r_1} \res \psi_i, \ldots, s_m \res \psi_m) $ \\ $
  = f(tgt(\psi_1), \ldots, tgt(\psi_i), \ldots, tgt(\psi_m)) 
  = tgt(\cfpceq{\psi}{f^\Box}) = tgt(\psi)
$.
\end{proof}

\medskip\noindent
\textbf{Lemma~\ref{rsl:proj-over-irs-peqb-if-respects}.} We recall the statement: 
\begin{quote}
Let $\psi$ be a proof term, $\mu: l \to h$ a rule symbol, and $r$ a position, such that $\psi$ respects $\set{r' \setsthat r' < r} \cup (r \cdot \PPos{\mu})$ and $\subtat{src(\psi)}{r} = l[s_1, \ldots, s_m]$.
Then
$\insertrs{src(\psi)}{\mu}{r} \comp \psi \wideres \insertrs{src(\psi)}{\mu}{r}
  \peqb \psi \comp \insertrs{tgt(\psi)}{\mu}{r} = \psi \comp \insertrs{src(\psi)}{\mu}{r} \wideres \psi$; \confer\ Lemma~\ref{rsl:res-step-over-respects}.
\end{quote}
\begin{proof}
We proceed by induction on $r$.

Assume that $r = \epsilon$, so that $src(\psi) = l[s_1, \ldots, s_m]$ and $\psi$ respects $\PPos{\mu}$. 
We have $\cfpceq{\psi}{l^\Box} = l[\psi_1, \ldots, \psi_m]$, so that recalling $\psi \peqb \cfpceq{\psi}{l^\Box}$ we obtain $src(\psi_i) = s_i$ for all $i$. 
Observe that $\insertrs{src(\psi)}{\mu}{r} = \mu(s_1, \ldots, s_m)$.
Clause \ref{it:proj-lmu} yields
$\psi \wideres \insertrs{src(\psi)}{\mu}{r} = h[\psi_1 \res s_1, \ldots, \psi_m \res s_m] = h[\psi_1, \ldots, \psi_m]$. Therefore, 
$\insertrs{src(\psi)}{\mu}{r} \comp \psi \wideres \insertrs{src(\psi)}{\mu}{r}
	= \mu(s_1, \ldots, s_m) \comp h[\psi_1, \ldots, \psi_m] 
	\peqb \mu(\psi_1, \ldots, \psi_m)
	\peqb l[\psi_1, \ldots, \psi_m] \comp \mu(tgt(\psi_1), \ldots, tgt(\psi_m))$,
applying \peqoutin\ and \peqinout\ resp. in the $\peqb$-steps.
We conclude by recalling that $l[\psi_1, \ldots, \psi_m] = \cfpceq{\psi}{l^\Box} \peqb \psi$, which in turn implies $tgt(\psi) = l[tgt(\psi_1), \ldots, tgt(\psi_m)]$.

Assume that $r = i r_1$. Say $src(\psi) = f(s_1, \ldots, s_m)$, observe that $f^\Box$ is a fixed prefix for $\psi$. We have $\cfpceq{\psi}{f^\Box} = f(\psi_1, \ldots, \psi_m)$, so that recalling $\psi \peqb \cfpceq{\psi}{f^\Box}$ we obtain $src(\psi_i) = s_i$ for all $i$, and also that $\cfpceq{\psi}{f^\Box}$ respects $\set{r' \sthat r' < r} \cup (r \cdot \PPos{\mu})$.
Observe $\insertrs{src(\psi)}{\mu}{r} = f(s_1, \ldots, \insertrs{s_i}{\mu}{r_1}, \ldots, s_m)$. Clause \ref{it:proj-ff} yields
$\psi \wideres \insertrs{src(\psi)}{\mu}{r} 
  = f(\psi_1 \res s_1, \ldots, \psi_i \res \insertrs{s_i}{\mu}{r_1}, \ldots, \psi_m \res s_m)
	= f(\psi_1, \ldots, \psi_i \res \insertrs{s_i}{\mu}{r_1}, \ldots, \psi_m)$.
On the other hand, $\psi_i$ respects $\set{r' \sthat r' < r_1} \cup (r_1 \cdot \PPos{\mu})$ and $\subtat{s_i}{r_1} = \subtat{src(\psi)}{r}$. Then \ih\ applies to $r_1$, yielding 
$\insertrs{s_i}{\mu}{r_1} \comp \psi_i \wideres \insertrs{s_i}{\mu}{r_1} \peqb \psi_i \comp \insertrs{tgt(\psi_i)}{\mu}{r_1}$.
Consequently, \\
$
\insertrs{src(\psi)}{\mu}{r} \comp \psi \wideres \insertrs{src(\psi)}{\mu}{r} 
\\ \hspace*{5mm}
\peqb
f(s_1, \ldots, \insertrs{s_i}{\mu}{r_1}, \ldots, s_m) \comp f(\psi_1, \ldots, \psi_i \res \insertrs{s_i}{\mu}{r_1}, \ldots, \psi_m)
\\ \hspace*{5mm}
\peqb
f(s_1 \comp \psi_1, \ldots, \insertrs{s_i}{\mu}{r_1} \comp \psi_i \res \insertrs{s_i}{\mu}{r_1}, \ldots, s_m \comp \psi_m)
\\ \hspace*{5mm}
\peqb
f(\psi_1 \comp tgt(\psi_1), \ldots, \psi_i \comp \insertrs{tgt(\psi_i)}{\mu}{r_1}, \ldots, \psi_m \comp tgt(\psi_m))
\\ \hspace*{5mm}
\peqb
f(\psi_1 , \ldots, \psi_i , \ldots, \psi_m ) \comp f(tgt(\psi_1), \ldots, \insertrs{tgt(\psi_i)}{\mu}{r_1}, \ldots, tgt(\psi_m))
$ \\
where structural equivalence, including \peqstruct\ and the easy fact $src(\chi) \comp \chi \peqb \chi \peqb \chi \comp tgt(\chi)$, is applied repeatedly.
Recalling that $\psi \peqb \cfpceq{\psi}{f^\Box}$ suffices to conclude, similarly as in the previous case.
\end{proof}

%% file: proj-comp-examples.tex
The clauses \ref{it:proj-comp-left} and \ref{it:proj-comp-right} in the definition of infinitary projection given in Section~\ref{sec:projection}, handle the compositions (either binary or infinite) that include head steps. Note that different cases involving compositions that do not include head steps match the clauses \ref{it:proj-mul}, \ref{it:proj-lmu} and \ref{it:proj-ff}.
When both $\phi$ and $\chi$ are compositions, and at least one of them includes a head step, in princple, either clause \ref{it:proj-comp-left} or clause \ref{it:proj-comp-right} could apply to $\phi \res \chi$. 
The added conditions on $\chi$ in the former clause describes the decision we have taken about this issue. We show through two examples, that the particular form of these conditions leads to a terminating (modulo computation of source/target of proof terms) computation of projections of infinite over finite, or finite over infinite, reductions in some cases.
The examples use only the rule $\mu: f(x) \to g(x)$.
Consider: \\[2pt]
$
\begin{array}{rcll}
\multicolumn{4}{l}{(\mu f\om \comp gf \mu f\om) \wideres (f \mu f\om \comp \icomp fgf g^i \mu f\om)} \arraysep
  & = & (\mu f\om \res (f \mu f\om \comp \icomp fgf g^i \mu f\om)) \comp (g f \mu f\om \wideres ((f \mu f\om \comp \icomp fgf g^i \mu f\om) \res \mu f\om))
	& \textnormal{cl. \ref{it:proj-comp-left}} \arraysep
  & = & \mu gf g\om \comp (g f \mu f\om \wideres ((f \mu f\om \comp \icomp fgf g^i \mu f\om) \res \mu f\om))
	& \textnormal{cl. \ref{it:proj-mul}, \ref{it:proj-id}} \arraysep
  & = & \mu gf g\om \comp (g f \mu f\om \wideres 
		((f \mu f\om \res \mu f\om) \comp \icomp fgf g^i \mu f\om \wideres (\mu f\om \res f \mu f\om)))
	& \textnormal{cl. \ref{it:proj-comp-left}} \arraysep
  & = & \mu gf g\om \comp (g f \mu f\om \wideres 
		(g \mu f\om \comp \icomp fgf g^i \mu f\om \res \mu g f\om)) \arraysep
  & = & \mu gf g\om \comp (g f \mu f\om \wideres 
		(g \mu f\om \comp g(\icomp gf g^i \mu f\om))) 
	& \textnormal{cl. \ref{it:proj-lmu}, \ref{it:proj-id}} \arraysep
  & = & \mu gf g\om \comp g (f \mu f\om \wideres 
		(\mu f\om \comp \icomp gf g^i \mu f\om) ) 
	& \textnormal{cl. \ref{it:proj-ff}} \arraysep
  & = & \mu gf g\om \comp g 
			( (f \mu f\om \res \mu f\om) \wideres \icomp gf g^i \mu f\om ) 
	& \textnormal{cl. \ref{it:proj-comp-right}} \arraysep
  & = & \mu gf g\om \comp g 
			( g \mu f\om \wideres \icomp gf g^i \mu f\om ) \arraysep
  & = & \mu gf g\om \comp g^2	(\mu f\om \wideres \icomp f g^i \mu f\om ) 
	& \textnormal{cl. \ref{it:proj-ff}} \arraysep
  & = & \mu gf g\om \comp g^2	\mu g\om
	& \textnormal{cl. \ref{it:proj-mul}, \ref{it:proj-id}} 
\end{array}
$ \\[2pt]
Note the use of clause \ref{it:proj-comp-left} in the first step, due to the condition $\chi = \icomp \chi_i$. Recall that $f \mu f\om \comp \icomp fgf g^i \mu f\om$ 
is in fact an infinite composition, whose first component is $f \mu f\om$, the second one is $fgf \mu f\om$, and so on.
Let us check how the projection computation would proceed if we used, in this situation, clause \ref{it:proj-comp-right} instead \\[2pt]
$
\begin{array}{rcll}
\multicolumn{4}{l}{(\mu f\om \comp gf \mu f\om) \wideres (f \mu f\om \comp \icomp fgf g^i \mu f\om)} \arraysep
  & = & ( (\mu f\om \comp gf \mu f\om) \res f \mu f\om ) \wideres \icomp fgf g^i \mu f\om
	& \textnormal{cl. \ref{it:proj-comp-right}} \arraysep
  & = & ( (\mu f\om \res f \mu f\om) \comp (gf \mu f\om \wideres (f \mu f\om \res \mu f\om) ) ) 
			\wideres \icomp fgf g^i \mu f\om
	& \textnormal{cl. \ref{it:proj-comp-left}} \arraysep
  & = & ( \mu g f\om  \comp (gf \mu f\om \wideres g \mu f\om ) ) 
			\wideres \icomp fgf g^i \mu f\om \arraysep
  & = & ( \mu g f\om  \comp g^2 \mu f\om ) \wideres \icomp fgf g^i \mu f\om
	& \textnormal{cl. \ref{it:proj-ff}, \ref{it:proj-lmu}, \ref{it:proj-id}} \arraysep
  & = & ( (\mu g f\om  \comp g^2 \mu f\om) \res fgf \mu f\om ) 
			\wideres \icomp fgf g^{i+1} \mu f\om
	& \textnormal{cl. \ref{it:proj-comp-right}} \arraysep
  & = & ( \mu g f g f\om \comp g^2 \mu f\om \res g^2 f \mu f\om )
			\wideres \icomp fgf g^{i+1} \mu f\om 
	& \textnormal{cl. \ref{it:proj-comp-left}} \arraysep
  & = & ( \mu g f g f\om \comp g^2 \mu g f\om )
			\wideres \icomp fgf g^{i+1} \mu f\om \arraysep
  & = & ( (\mu g f g f\om \comp g^2 \mu g f\om) \res fgfg \mu f\om )
			\wideres \icomp fgf g^{i+2} \mu f\om 
	& \textnormal{cl. \ref{it:proj-comp-right}} \arraysep
	& & \ldots
\end{array}
$ \\[2pt]
This computation would continue indefinitely, since it loops over projections of the form $(\phi_1 \comp \phi_2) \wideres \psi$, where $\phi_1$ and $\phi_2$ are one-steps, and $\psi$ is an infinite composition.

\medskip\noindent
Let us consider now an infinite-over-finite case: \\[2pt]
$
\begin{array}{rcll}
\multicolumn{4}{l}{(f \mu f\om \comp \icomp fgf g^i \mu f\om) \wideres (\mu f\om \comp gf \mu f\om)} \arraysep
  & = & ((f \mu f\om \comp \icomp fgf g^i \mu f\om) \res \mu f\om) \wideres gf \mu f\om
			& \textnormal{cl. \ref{it:proj-comp-right}} \arraysep
  & = & ( (f \mu f\om \res \mu f\om) \comp \icomp fgf g^i \mu f\om \wideres (\mu f\om \res f \mu f\om) )
	\wideres gf \mu f\om
			& \textnormal{cl. \ref{it:proj-comp-left}} \arraysep
  & = & ( g \mu f\om \comp \icomp fgf g^i \mu f\om \res \mu g f\om )
	\wideres gf \mu f\om \arraysep
  & = & ( g \mu f\om \comp g (\icomp gf g^i \mu f\om) )
	\wideres gf \mu f\om \arraysep
  & = & g ( (\mu f\om \comp \icomp gf g^i \mu f\om) \wideres f \mu f\om ) 
			& \textnormal{cl. \ref{it:proj-ff}} \arraysep
  & = & g ( 
	(\mu f\om \res f \mu f\om ) \comp \icomp gf g^i \mu f\om \wideres (f \mu f\om \res \mu f\om)
					) 
			& \textnormal{cl. \ref{it:proj-comp-left}} \arraysep
  & = & g ( \mu g f\om \comp \icomp gf g^i \mu f\om \wideres g \mu f\om ) \arraysep
  & = & g ( \mu g f\om \comp g^2 (\icomp g^i \mu f\om) ) 
	& \textnormal{cl. \ref{it:proj-ff}, \ref{it:proj-lmu}, \ref{it:proj-id}} \arraysep
  & \peqb & g \mu g f\om \comp g^3 (\icomp g^i \mu f\om) 
\end{array}
$ \\[2pt]
We observe that clause \ref{it:proj-comp-right} applies in the first step, since $\mu f\om \comp gf \mu f\om$ is a binary composition. The use of clause \ref{it:proj-comp-left} in such a case would lead to \\[2pt]
$
\begin{array}{rcll}
\multicolumn{4}{l}{(f \mu f\om \comp \icomp fgf g^i \mu f\om) \wideres (\mu f\om \comp gf \mu f\om)} \arraysep
  & = & 
    (f \mu f\om \wideres (\mu f\om \comp gf \mu f\om))
    \comp
    \icomp fgf g^i \mu f\om \wideres ((\mu f\om \comp gf \mu f\om) \wideres f \mu f\om)
      & \ \ \ \textnormal{cl. \ref{it:proj-comp-left}} \arraysep
  & = & 
    g \mu g f\om
    \comp
    \icomp fgf g^i \mu f\om \wideres (\mu g f\om \comp g^2 \mu f\om) \arraysep
  & & \ldots
\end{array}
$ \\[2pt]
We find again a loop, now on projections of the form $\psi \wideres (\phi_1 \comp \phi_2)$.

\bigskip
These observations lead to the precise form of the definition of projections we propose in this article.

%% file: projections-infinitary-rewriting-lombardi-rios-devrijer.bbl
\begin{thebibliography}{10}

\bibitem{nosotrosPopl2014}
B.~Accattoli, E.~Bonelli, D.~Kesner, and C.~Lombardi.
\newblock A nonstandard standardization theorem.
\newblock In S.~Jagannathan and P.~Sewell, editors, {\em POPL}, pages 659--670.
  ACM, 2014.

\bibitem{barendregt}
H.P. Barendregt.
\newblock {\em The Lambda Calculus: Its Syntax and Semantics}.
\newblock Elsevier, Amsterdam, 1984.

\bibitem{curry-feys}
H.~B. Curry and R.~Feys.
\newblock {\em Combinatory Logic}.
\newblock North-Holland Publishing Company, Amsterdam, 1958.

\bibitem{rewriterewrite}
N.~Dershowitz, S.~Kaplan, and D.~Plaisted.
\newblock Rewrite, rewrite, rewrite, rewrite, rewrite, . .
\newblock {\em Theor. Comput. Sci.}, 83(1):71--96, 1991.

\bibitem{endrullis-communication-2016}
J.~Endrullis.
\newblock Personal communication, 2016.

\bibitem{endrullis-rta-2015}
J.~Endrullis, H.~Hvid Hansen, D.~Hendriks, A.~Polonsky, and A.~Silva.
\newblock A coinductive framework for infinitary rewriting and equational
  reasoning.
\newblock In M.~Fern{\'{a}}ndez, editor, {\em {RTA} 2015}, volume~36 of {\em
  LIPIcs}, pages 143--159. Schloss Dagstuhl - Leibniz-Zentrum fuer Informatik,
  2015.

\bibitem{inf-ars}
R.~Kennaway.
\newblock On transfinite abstract reduction systems.
\newblock Technical Report CS-R9205, Centrum voor Wiskunde en Informatica,
  Netherlands, 1992.

\bibitem{orthogonal-itrs-95}
R.~Kennaway, J.W. Klop, M.~Ronan Sleep, and F.-J. de~Vries.
\newblock Transfinite reductions in orthogonal term rewriting systems.
\newblock {\em Inf. Comput.}, 119(1):18--38, 1995.

\bibitem{KetemaRTA12}
J.~Ketema.
\newblock Reinterpreting compression in infinitary rewriting.
\newblock In A.~Tiwari, editor, {\em RTA 2012 (Nagoya, Japan)}, volume~15 of
  {\em LIPIcs}, pages 209--224. Schloss Dagstuhl - Leibniz-Zentrum fuer
  Informatik, 2012.

\bibitem{inf-normalization}
J.W. Klop and R.~de~Vrijer.
\newblock Infinitary normalization.
\newblock In {\em We Will Show Them: Essays in Honour of Dov Gabbay}, volume~2,
  pages 169--192. College Publications, 2005.

\bibitem{phdcarlos}
C.~Lombardi.
\newblock {\em Reduction spaces in non-sequential and infinitary rewriting
  systems}.
\newblock Phd thesis, Universidad de Buenos Aires -- Universit\'e
  Paris-Diderot, 2014.

\bibitem{nosotros-rta14}
C.~Lombardi, A.~R\'{\i}os, and R.~de~Vrijer.
\newblock Proof terms for infinitary rewriting.
\newblock In G.~Dowek, editor, {\em RTA-TLCA'14}, volume 8560 of {\em Lecture
  Notes in Computer Science}, pages 303--318. Springer, 2014.

\bibitem{long-version}
C.~Lombardi, A.~R\'{\i}os, and R.~de~Vrijer.
\newblock Projections for infinitary rewriting.
\newblock Online at \url{http://arxiv.org/abs/1605.07808}, 2016.

\bibitem{thesis-mellies}
P.-A. Melli\`es.
\newblock {\em Description abstraite des Syst\`emes de R\'e\'ecriture}.
\newblock PhD thesis, Univ. Paris VII, 1996.

\bibitem{terese}
Terese.
\newblock {\em Term Rewriting Systems}, volume~55 of {\em Cambridge Tracts in
  Theoretical Computer Science}.
\newblock Cambridge University Press, Cambridge, UK, 2003.

\end{thebibliography}
